\documentclass[11pt]{article}

\usepackage{titlesec} 
\titleformat*{\section}{\large\bfseries}
\titleformat*{\subsection}{\large\bfseries}

\usepackage{hyperref} 
\usepackage[english]{babel}
\usepackage{mathtools, amssymb} 
\usepackage{mathtext} 
\usepackage{bm} 
\usepackage{amsthm} 
\usepackage{dsfont}
\usepackage{tikz}
\usetikzlibrary{patterns,arrows}

\usepackage{tikz, tikz-cd,subcaption, float}
\usetikzlibrary{decorations.markings, arrows.meta}
\usepackage[ddmmyy]{datetime}

\usepackage{tikz}\usetikzlibrary{
	arrows,
	arrows.meta,
	bending,
	calc,
	patterns,
	intersections,
	decorations.pathreplacing,
	decorations.pathmorphing,
	decorations.text,
	decorations.markings,
	shapes,
	positioning,
	through,
	quotes,
	angles
}

\definecolor{spec}{RGB}{0,50,230}

\tikzset{wavy/.style={decorate, decoration={snake, segment length=4mm, amplitude=0.7mm}}}
\tikzset{mydash/.style={dash pattern=on 2pt off 1.5pt}}

\tikzset{->-/.style={line width=0.8pt,decoration={
			markings,
			mark=at position 0.56 with {\arrow{Stealth[length=1.7mm,width=1.3mm]}}},postaction={decorate}}}

\usepackage{graphicx}
\usepackage{epsfig,epsf}
\usepackage{pdfsync}

\hypersetup{colorlinks=true, linktocpage=false, linkcolor=blue!80!black, citecolor=red!85!black, urlcolor=blue} 

\theoremstyle{plain}
\newtheorem{theorem}{Theorem}[section]
\newtheorem{proposition}[theorem]{Proposition}
\newtheorem{corollary}[theorem]{Corollary}
\theoremstyle{definition}
\newtheorem{definition}[theorem]{Definition}

\theoremstyle{remark}
\newtheorem{remark}[theorem]{Remark}

\tikzset{->-/.style={line width=0.25mm,decoration={
			markings,
			mark=at position 0.53 with {\arrow{Stealth[length=1.4mm,width=1.1mm]}}},postaction={decorate}}}

\renewcommand{\imath}{\mathrm{i}}

\renewcommand{\phi}{\varphi}

\def\II{\hbox{{1}\kern-.25em\hbox{l}}}
\def\tr{\mathop{\rm tr}\nolimits}

\usepackage{tocloft}

\numberwithin{equation}{section}
\title{Quasideterminants and differential realizations of W-algebras}
\author{Sergey Derkachov$^{1,3,4}$ and Daniil Getta$^{2,3,5}$}
\date{}
\begin{document}

\maketitle
\thispagestyle{empty}

\begin{center}
{\small \it
		$^1$ Steklov Mathematical Institute, Fontanka 27, \\
St.~Petersburg, 191023, Russia \vspace{0.2cm}\\$^2$Skolkovo Institute of Science and Technology, \\Moscow, 121205, Russia\\
		$^3$ St. Petersburg State University, St. Petersburg,\\ 199034, Russia\\
        $^4$ Beijing Institute of Mathematical Sciences and Applications, \\
		Huairou district, Beijing, 101408, China ,\\
        $^5$ National Research University Higher School of Economics, \\Moscow, 101000, Myasnitskaya ul., 20
        }
\end{center}
\begin{abstract}
In this work, we study the Hamiltonian reduction, used in constructing
finite rectangular $W$-algebras, of the big-cell differential realization of
$\mathfrak{gl}_{np}$. We establish an intrinsic
factorization of the reduced differential realization, identifying its
Yangian action with the coproduct action on a tensor product of $n$
evaluation differential realizations of $\mathfrak{gl}_p$. Corner
quasideterminants provide the essential noncommutative mechanism through
which Hamiltonian reduction transforms the higher-rank differential
realization into a factorized lower-rank Yangian action.
\end{abstract}

\newpage
\tableofcontents \newpage
\section{Introduction}\label{sec:introduction}

$W$-algebras are now central objects in modern mathematics and mathematical
physics. Since their introduction by Zamolodchikov
\cite{Zamolodchikov}, they have found important applications
to integrable partial differential equations and Drinfeld--Sokolov
hierarchies \cite{DS}, integrable field theories
\cite{FateevLitvinov}, and two-dimensional conformal field theory
\cite{DotsenkoFateev}. Later developments revealed that $W$-algebras also
form an essential part of the correspondence between four-dimensional
supersymmetric gauge theories and two-dimensional conformal field theories
\cite{AGT}. 

One of the principal ways to construct $W$-algebras is Hamiltonian
reduction, known in the affine setting as Drinfeld--Sokolov reduction
\cite{DS,FeiginFrenkel}. Quantum Hamiltonian reduction also produces
finite $W$-algebras associated with nilpotent
elements and compatible good gradings. For $\mathfrak{gl}_N$, the relevant
data can be encoded by pyramids \cite{ElashviliKac,BK}. In this work we
consider the finite $W$-algebras for $\mathfrak{gl}_N$ associated with a
rectangular pyramid.

Let $N=np$. The algebra $\mathfrak g=\mathfrak{gl}_N$ then has a natural
$n\times n$ block structure with blocks of size $p\times p$. We denote its
matrix units by $E_{a,b}^{i,j}$, where $a,b\in\{1,\ldots,n\}$ label the
blocks and $i,j\in\{1,\ldots,p\}$ label the entries inside a block. The
subalgebra used in the reduction is
$$
\mathfrak m_n
=\operatorname{span}\{E_{a,b}^{i,j}\mid a>b\},
$$
and the character is
$
\psi(E_{a,b}^{i,j})=\delta_{i,j}\delta_{a,b+1}.
$
We introduce the left ideal
$$
I_\psi=\mathcal U(\mathfrak{gl}_N)(\mathfrak m_n-\psi)
$$
and define the rectangular finite $W$-algebra by
$$
\mathcal W_n(\mathfrak{gl}_p)
=\bigl(\mathcal U(\mathfrak{gl}_N)/I_\psi\bigr)^{\mathfrak m_n}.
$$

A deep connection between these finite $W$-algebras and Yangians was
established in \cite{RS,BK}. In the rectangular case, Brundan and
Kleshchev identify the finite $W$-algebra with a truncated Yangian.

\begin{theorem}[Brundan--Kleshchev \cite{BK}]
\label{thm:intro_truncated_yangian_w_alg}
The finite $W$-algebra $\mathcal W_n(\mathfrak{gl}_p)$ is isomorphic to
the level-$n$ truncated Yangian $Y^{(n)}(\mathfrak{gl}_p)$.
\end{theorem}

Heuristically, one may view this result as passing from a one-site
spin-chain object of rank $np$ to an $n$-site generating matrix with
rank-$p$ symmetry by Hamiltonian reduction. This picture is only
heuristic, but it helps explain why the Yangian coproduct and the
factorized form of the generating matrix arise naturally.

Quasideterminants play a central role in the explicit construction of the
generating matrix. De Sole, Kac, and Valeri \cite{DSKV} used them to
describe the isomorphism with the truncated Yangian, building on the
general theory developed in \cite{GGRW}. Unlike an ordinary determinant,
a quasideterminant retains the order of multiplication of its
noncommuting entries, making it well suited to constructions in
noncommutative algebras.

More precisely, the corner quasideterminant of the $L$-operator, a matrix
with entries in $\mathcal{U}(\mathfrak{gl}_N)[u]$, is invariant under the
action of the subalgebra used in the reduction and satisfies the RTT
relation. Its image under reduction is the Brundan--Kleshchev generating
matrix, so both the RTT relation and the generating property pass to the
finite $W$-algebra. Thus the quasideterminant constructs the reduction
invariant while carrying its Yangian structure through the reduction.

Another important tool in representation theory and mathematical physics
is the differential realization of a Lie algebra on the big cell of its
associated flag variety \cite{Kostant,BeilinsonBernstein}:
\begin{equation}\label{eq:quantum_Springer_map}
\rho_\lambda:
\mathcal U(\mathfrak g)\longrightarrow\mathcal D(N_-),
\end{equation}
where $N_-\subset\mathrm{GL}_N$ is the lower unipotent subgroup,
$\mathcal D(N_-)$ is its algebra of polynomial differential operators,
and $\lambda\in\mathfrak h^*$ is a weight parameter. These realizations
give explicit differential-operator models of representations and
provide a natural setting in which to study reduction. The affine
analogues of these realizations lead to Wakimoto realizations of affine
Lie algebras \cite{Witten,deBoerFeher}.

Together, these constructions lead to the central question of this paper:

\emph{What is the Yangian structure induced by quantum Hamiltonian reduction of the big-cell differential realization}

A related but separate question is how modules and intertwiners behave
under reduction. We plan to address it in the continuation of this work.

We answer the first question by proving that the Yangian action induced by
Hamiltonian reduction of the big-cell differential realization of
$\mathfrak{gl}_N$ coincides with the coproduct action on a tensor product
of $n$ rank-$p$ differential realizations. The key step is an explicit
factorization of the reduced quasideterminant. This establishes
compatibility between two constructions
that arise from different perspectives: Hamiltonian reduction of the
higher-rank differential realization and the coproduct realization of the
Yangian on a tensor product of lower-rank differential realizations.
The quasideterminant streamlines the calculations and makes transparent
the identities and invariance properties needed for the finite
$W$-algebra generators.

To state the result precisely, let $N_-'\subset\mathrm{GL}_p$ be the
lower unipotent subgroup. The reduced differential algebra is naturally
identified with
\begin{align}
\xi:
\mathcal D_-
\xrightarrow{\sim}
\mathcal D(N_-')^{\otimes n},
\end{align}
where $\mathcal D_-$ denotes the differential algebra obtained by
Hamiltonian reduction of $\mathcal D(N_-)$. Via the
Brundan--Kleshchev identification, this reduction endows
$\mathcal D(N_-')^{\otimes n}$ with a $Y(\mathfrak{gl}_p)$-module
structure. The generating matrix for this action is the
Brundan--Kleshchev matrix evaluated in the reduced differential
realization.

There is also a direct construction of a Yangian action on
$\mathcal D(N_-')^{\otimes n}$, using the coproduct together with
evaluation homomorphisms. Write
$$
\lambda=\operatorname{diag}
(\boldsymbol\lambda_1,\ldots,\boldsymbol\lambda_n),
$$
where each $\boldsymbol\lambda_a$ is a diagonal $p\times p$ matrix. These
parameters define differential realizations
$$
\rho_{\boldsymbol\lambda_a}:
\mathcal U(\mathfrak{gl}_p)\longrightarrow\mathcal D(N_-'),
$$
one in each tensor factor. Let $E_{[a]}$ denote the generator matrix in
the $a$th factor after applying the corresponding realization; we suppress
the homomorphisms in operator formulas. In polynomial normalization, the
associated generating matrix is
$$
\mathrm{T}(u)=
\bigl(u+E_{[n]}\bigr)
\cdots
\bigl(u+E_{[1]}\bigr).
$$

\begin{theorem}\label{thm:intro_main_factorization}
After the identification by $\xi$, the Hamiltonian reduction of
\eqref{eq:quantum_Springer_map} and the Yangian action induced by the
coproduct and the tensor product of differential realizations define the
same $Y(\mathfrak{gl}_p)$-module structure on
$\mathcal D(N_-')^{\otimes n}$.
\end{theorem}

The proof uses the invariance of the corner quasideterminant under certain
block multiplications. These multiplications remove the off-diagonal
coordinate factors and leave a matrix whose entries are already invariant
under the reduction. The reduction can therefore be performed on these
entries inside the quasideterminant, yielding the ordered product asserted
in the theorem.

The factorization also motivates the representation-theoretic continuation
of this work, in which we will study reduced intertwiners. Let $W$ and
$W'$ be Whittaker modules obtained from suitable specializations of the
differential realization. Write $\operatorname{Wh}(W)$ for the subspace of
Whittaker vectors in $W$, namely, vectors $v$ satisfying
$$
E_{a,b}^{i,j}v=\psi(E_{a,b}^{i,j})v.
$$
Skryabin's equivalence \cite{Skryabin} gives an isomorphism between the
corresponding intertwining spaces:
$$
\operatorname{Hom}_{\mathfrak{gl}_N}(W,W')
\simeq
\operatorname{Hom}_{\mathcal W_n(\mathfrak{gl}_p)}
\bigl(\operatorname{Wh}(W),\operatorname{Wh}(W')\bigr).
$$
The factorization established here gives an explicit functional model for
the Whittaker subspaces in terms of functions on $n$ copies of the
lower-rank group $N_-'$. Upon reduction, natural intertwiners of the
unreduced differential realizations induce intertwiners between these
Whittaker spaces. These operators are closely related to the factorized
$R$-matrices and Baxter $Q$-operators studied by Derkachov and Manashov
\cite{DM}. We plan to develop this connection in the second part of the
work.

The paper is organized as follows. In
Section~\ref{sec:yangian_preliminaries} we introduce the Yangian and the
truncated Yangian, recall the coproduct and evaluation homomorphisms, and
review the tensor-product realization. In
Section~\ref{sec:quasideterminants} we introduce the corner
quasideterminant, relate it to a block of the inverse matrix, and prove the
RTT relation and commutation properties that make it an invariant for
reduction. In Section~\ref{sec:W_algebras} we recall rectangular finite
$W$-algebras, identify the reduced quasideterminant with the
Brundan--Kleshchev generating matrix, and describe the Miura
homomorphism.

In Section~\ref{sec:differential_reduction} we review the differential
realization on the big cell, identify its reduced differential algebra
with $\mathcal D(N_-')^{\otimes n}$, and compute the reduction of the
quasideterminant to obtain the factorization in the main theorem.
Appendix~\ref{appA} collects auxiliary calculations. Its first subsection
gives a direct proof of the commutation relations, and the second provides
detailed index calculations used in Section~3. Appendix~\ref{appB} proves
the block proposition and explains the reduction in a concrete functional
realization.

\section{Preliminaries on the RTT Yangian}
\label{sec:yangian_preliminaries}

Let $e_{i,j}$, $1\leq i,j\leq N$, be the standard matrix units in
$\operatorname{Mat}_N(\mathbb C)$, and let $E_{i,j}$ be the standard
generators of the Lie algebra $\mathfrak{gl}_N$, with commutation relations
\begin{align}\label{eq:comm_relations_gl}[E_{i,j},E_{k,q}]=\delta_{j,k}E_{i,q}-\delta_{q,i}E_{k,j}.\end{align}
Let
$
P=\sum_{i,j=1}^N e_{i,j}\otimes e_{j,i}$
be the permutation operator on $\mathbb C^N\otimes\mathbb C^N$. The
rational $R$-matrix
\begin{equation}\label{eq:rational_R_matrix}
R(u)=u+P \in \mathrm{End}(\mathbb{C}^N\otimes \mathbb{C}^N)[u]
\end{equation}
satisfies the Yang--Baxter equation
\begin{equation}\label{eq:rational_Yang_Baxter}
R_{12}(u)R_{13}(u+v)R_{23}(v)
=R_{23}(v)R_{13}(u+v)R_{12}(u),
\end{equation}
where $R_{i,j}$ indicates the two tensor factors on which $R$ acts
nontrivially.

\begin{definition}\label{def:RTT_Yangian}
The RTT Yangian $Y(\mathfrak{gl}_N)$ is the unital
associative algebra generated by
$$
t_{i,j}^{(r)},\qquad 1\leq i,j\leq N,\quad r\geq1,
$$
whose generating series
\begin{equation}\label{eq:unshifted_RTT_matrix}
t_{i,j}(u)=\delta_{i,j}+\sum_{r\geq1}t_{i,j}^{(r)}u^{-r},
\qquad
\mathrm{T}(u)=\sum_{i,j=1}^N e_{i,j}\otimes t_{i,j}(u),
\end{equation}
satisfy the relations
\begin{align} \label{eq:yangian_relation}
(u-v)t_{i,j}(u)t_{k,q}(v)+t_{k,j}(u)t_{i,q}(v)=
(u-v)t_{k,q}(v)t_{i,j}(u)+t_{k,j}(v)t_{i,q}(u),
\end{align}
where $\mathrm{T}_1(u)=\mathrm{T}(u)\otimes \mathbf{1}_N$ and
$\mathrm{T}_2(v)=\mathbf{1}_N \otimes\mathrm{T}(v)$ are
$N^2\times N^2$ matrices with entries in
$Y(\mathfrak{gl}_N)[[u^{-1},v^{-1}]]$.
\end{definition}
We next recall the coproduct of the Yangian.

\begin{proposition}\label{prp:RTT_coproduct}
There is an algebra homomorphism
$$
\Delta:Y(\mathfrak{gl}_N)\longrightarrow
Y(\mathfrak{gl}_N)\otimes Y(\mathfrak{gl}_N)
$$
determined by
\begin{equation}\label{eq:RTT_coproduct_matrix}
\Delta(\mathrm{T}(u))=\mathrm{T}_{[1]}(u)\mathrm{T}_{[2]}(u),
\end{equation}
where $\mathrm{T}_{[1]}(u)=\mathrm{T}(u)\otimes 1$ and
$\mathrm{T}_{[2]}(u)=1\otimes\mathrm{T}(u)$ are elements of
$\mathrm{Mat}_N(\mathbb C)\otimes
Y(\mathfrak{gl}_N)^{\otimes2}[[u^{-1}]]$.
Equivalently,
$$
\Delta(t_{i,j}(u))=\sum_{k=1}^N
t_{i,k}(u)\otimes t_{k,j}(u).
$$
The map $\Delta$ is coassociative:
$$
(\Delta\otimes1)\circ\Delta=(1\otimes\Delta)\circ\Delta.
$$
\end{proposition}

Multiplication of the monodromy matrix by a scalar power of $u$ preserves
the RTT relation. We denote by $Y_m(\mathfrak{gl}_N)$ the corresponding
polynomial normalization, in which
$\mathrm{T}(u)=u^m+O(u^{m-1})$. The map
\begin{equation}\label{eq:Yangian_with_different_n}
\mathrm{T}(u)\longmapsto u^m\mathrm{T}(u)
\end{equation}
identifies $Y_0(\mathfrak{gl}_N)$ with $Y_m(\mathfrak{gl}_N)$. The coproduct
is compatible with this normalization and defines
$$
\Delta_{m_1,m_2}:Y_{m_1+m_2}(\mathfrak{gl}_N)\longrightarrow
Y_{m_1}(\mathfrak{gl}_N)\otimes Y_{m_2}(\mathfrak{gl}_N).
$$

\begin{proposition}
For every $c\in\mathbb C$, the assignment
$$
\operatorname{ev}_c(\mathrm{T}(u))
=1+\frac{1}{u-c}\sum_{i,j=1}^N e_{i,j}\otimes E_{j,i}
$$
defines an evaluation homomorphism
$$
\operatorname{ev}_c:Y(\mathfrak{gl}_N)\longrightarrow
\mathcal U(\mathfrak{gl}_N).
$$
In the polynomial normalization, the same homomorphism is written as
$$
\operatorname{ev}_c:Y_1(\mathfrak{gl}_N)\longrightarrow
\mathcal U(\mathfrak{gl}_N),
\qquad
\operatorname{ev}_c(\mathrm{T}(u))=L(u-c),
\qquad
L(u)=u+\sum_{i,j=1}^N e_{i,j}\otimes E_{j,i}.
$$
\end{proposition}

Thus every $\mathfrak{gl}_N$-module $M$ becomes an evaluation
$Y_1(\mathfrak{gl}_N)$-module, denoted by $M(c)$. The coproduct equips
$$
M_1(z_1)\otimes\cdots\otimes M_m(z_m)
$$
with a $Y_m(\mathfrak{gl}_N)$-module structure. Its monodromy matrix is
\begin{equation}\label{eq:iterated_evaluation}
(\operatorname{ev}_{z_1}\otimes\cdots\otimes\operatorname{ev}_{z_m})
\circ\Delta_{1,\dots,1}(\mathrm{T}(u))
=L_{[1]}(u-z_1)\cdots L_{[m]}(u-z_m),
\end{equation}
where $\Delta_{1,\dots,1}$ is the iterated coproduct and $L_{[i]}(u)$ is
the $L$-operator acting nontrivially only in the $i$-th tensor factor.

\begin{definition}
In $Y_m(\mathfrak{gl}_N)$, we reindex the coefficients by writing
$$
t_{i,j}(u)=\delta_{i,j}u^m+t_{i,j}^{(1-m)}u^{m-1}+...+t^{(0)}_{i,j}+
\sum_{r> 0}t_{i,j}^{(r)}u^{-r}.
$$
For $m\in\mathbb N$, the level-$m$ truncated Yangian is the quotient
$$
Y^{(m)}(\mathfrak{gl}_N)
=Y_m(\mathfrak{gl}_N)\big/
\bigl(t_{i,j}^{(r)}\mid r>0,\ 1\leq i,j\leq N\bigr).
$$
\end{definition}
This quotient is the Yangian of level $m$ introduced by Cherednik
\cite{Ch1,Ch2}. The following tensor-product realization is appeared 
naturally in the framework of the Quantum Inverse Scattering Method 
\cite{FST,KuSk1,Skl,Fad,KRS}.

\begin{proposition}[Cherednik \cite{Ch1,Ch2}; see also \cite{BK}]
\label{th:Cherednik_prp}
For $z_1=\cdots=z_m=0$, the homomorphism
\eqref{eq:iterated_evaluation} descends to an injective homomorphism
$$
Y^{(m)}(\mathfrak{gl}_N)\hookrightarrow
\mathcal U(\mathfrak{gl}_N)^{\otimes m}.
$$
Its image is generated by the matrix coefficients of the polynomial
\begin{align} \label{eq:quantum_yangian_generating_matrix}
\mathrm{T}(u)=L_{[1]}(u)L_{[2]}(u)\cdots L_{[m]}(u).
\end{align}
\end{proposition}

\section{Quasideterminants and commutation relations}
\label{sec:quasideterminants}

In this section, we collect the quasideterminant identities used below and
prove the two properties of the corner quasideterminant that are essential
for Hamiltonian reduction: its RTT relation and its commutativity with the
blocks belonging to the reduction subalgebra. We keep the block
calculations explicit, since the positions of the blocks are important in
both arguments.

\subsection{Basic notation}

Fix positive integers $n,p$, set $N=np$, and let $E_{r,s}$ denote the
standard generators of $\mathfrak{gl}_N$. We use the block notation
$$
E_{a,b}^{i,j}
:=E_{(a-1)p+i,\,(b-1)p+j}.
$$
The lower indices record the $n\times n$ block position, while the upper
indices record the matrix entry inside a $p\times p$ block. The
commutation relations \eqref{eq:comm_relations_gl} can be written as
\begin{equation}\label{eq:generator_commutation_relations}
[E_{a,b}^{i,j},E_{c,d}^{k,r}]
=\delta_{b,c}\delta_{j,k}E_{a,d}^{i,r}
-\delta_{d,a}\delta_{r,i}E_{c,b}^{k,j}.
\end{equation}
\begin{definition}
Introduce the universal matrix
\begin{equation}\label{eq:universal_matrix}
E=\sum_{a,b=1}^n\sum_{i,j=1}^p
e_{(a-1)p+i,\,(b-1)p+j}\otimes E_{b,a}^{j,i}
\in \operatorname{Mat}_N(\mathcal U(\mathfrak{gl}_N)).
\end{equation}
The quantum $L$-operator is
$$L(u)=u+E\in \mathrm{Mat}_{N}(\mathcal{U}(\mathfrak{gl}_N)[u]).$$
\end{definition}
We write the matrix blocks in boldface.
For $1\leq a,b\leq n$, set
\begin{equation}\label{eq:generator_blocks}
\mathbf E_{a,b}
=\sum_{i,j=1}^pe_{(a-1)p+i,(b-1)p+j}\otimes E_{b,a}^{j,i}.
\end{equation}
In $n\times n$ block form,
\begin{equation}\label{eq:block_L_operator}
L(u)=\begin{pmatrix}
u+\mathbf E_{1,1}&\mathbf E_{1,2}&\cdots&\mathbf E_{1,n}\\
\mathbf E_{2,1}&u+\mathbf E_{2,2}&\cdots&\mathbf E_{2,n}\\
\vdots&\vdots&\ddots&\vdots\\
\mathbf E_{n,1}&\mathbf E_{n,2}&\cdots&u+\mathbf E_{n,n}
\end{pmatrix},
\end{equation}
where every block is $p\times p$ and the scalar $u$ on each diagonal
block denotes scalar multiplication by the identity matrix.

\paragraph{Example.}
For $N=4$ and $n=p=2$, the same $L$-operator can be written successively
in ordinary matrix notation, indexed block notation, and block-matrix
notation:
\begin{align*}
L(u)= \left(\begin{array}{cccc}
u+\mathrm{E}_{1,1} & \mathrm{E}_{2,1} & \mathrm{E}_{3,1} & \mathrm{E}_{4,1}\\
\mathrm{E}_{1,2} & u+\mathrm{E}_{2,2} & \mathrm{E}_{3,2} & \mathrm{E}_{4,2}\\
\mathrm{E}_{1,3} & \mathrm{E}_{2,3} & u+\mathrm{E}_{3,3}& \mathrm{E}_{4,3}\\
\mathrm{E}_{1,4} & \mathrm{E}_{2,4} & \mathrm{E}_{3,4}& u+\mathrm{E}_{4,4}
\end{array}\right) &=
\left(\begin{array}{cccc}
u+E_{1,1}^{1,1}&E_{1,1}^{2,1}
&E_{2,1}^{1,1}&E_{2,1}^{2,1}\\
E_{1,1}^{1,2}&u+E_{1,1}^{2,2}
&E_{2,1}^{1,2}&E_{2,1}^{2,2}\\
E_{1,2}^{1,1}&E_{1,2}^{2,1}
&u+E_{2,2}^{1,1}&E_{2,2}^{2,1}\\
E_{1,2}^{1,2}&E_{1,2}^{2,2}
&E_{2,2}^{1,2}&u+E_{2,2}^{2,2}
\end{array}\right)\\ & = \left(\begin{array}{cc}
u+\mathbf E_{1,1}&\mathbf E_{1,2}\\
\mathbf E_{2,1}&u+\mathbf E_{2,2}\end{array}\right).
\end{align*}

\subsection{Quasideterminants, the inverse matrix, and the row determinant}

We recall the definition of a quasideterminant \cite{GR1,GR2,GGRW} and the formulas
needed below. Let $A$ be a matrix with entries in a possibly
noncommutative ring $\mathcal A$:
\begin{align}
A = \begin{pmatrix}
a_{1,1} & a_{1,2} &\dots & a_{1,n} \\
\vdots & \vdots &\ddots & \vdots \\
a_{n-1,1} & a_{n-1,2}&\dots& a_{n-1,n}\\
a_{n,1} & a_{n,2} & \dots & a_{n,n}
\end{pmatrix}.
\end{align}
\begin{definition}\label{def:corner_quasideterminant}
The corner quasideterminant of $A$ with respect to the entry $a_{n,1}$ is
\begin{align}
|A|_{n,1} = a_{n,1} -
\begin{pmatrix} a_{n,2} & \dots & a_{n,n} \end{pmatrix}
\begin{pmatrix}
a_{1,2} &\dots & a_{1,n} \\
\vdots &\ddots & \vdots \\
a_{n-1,1} &\dots & a_{n-1,n}
\end{pmatrix}^{-1}
\begin{pmatrix} a_{1,1}\\ \vdots \\ a_{n-1,1}\end{pmatrix}
\end{align}
provided that the displayed $(n-1)\times(n-1)$ submatrix is invertible in
$\mathcal A$. 
\end{definition}

\begin{proposition}[\cite{GGRW}]\label{prp:quasidet_and_inverse}
Assume that $A^{-1}$ and $|A|_{n,1}^{-1}$ exist over $\mathcal A$. Then
    \begin{align}\label{eq:inverse_matrix_quasidet}
|A|_{n,1}^{-1}= (A^{-1})_{1,n}.
    \end{align}

\end{proposition}
\paragraph{Example.}
For $n=2$, the corner quasideterminant is
\begin{align*}
A = \begin{pmatrix} a_{1,1} & a_{1,2}\\
a_{2,1} & a_{2,2} \end{pmatrix} \ \ ;\ \
\begin{aligned}
|A|_{2,1} = a_{2,1}-a_{2,2} a^{-1}_{1,2} a_{1,1}.
\end{aligned}
\end{align*}
If, in addition, $A$ and all the quasideterminants appearing below are
invertible over $\mathcal A$, then
\begin{align*}
A^{-1}=\begin{pmatrix}
(a_{1,1}-a_{1,2}a_{2,2}^{-1}a_{2,1})^{-1}
& (a_{2,1}-a_{2,2}a_{1,2}^{-1}a_{1,1})^{-1}\\[1ex]
(a_{1,2}-a_{1,1}a_{2,1}^{-1}a_{2,2})^{-1}
& (a_{2,2}-a_{2,1}a_{1,1}^{-1}a_{1,2})^{-1}
\end{pmatrix}.
\end{align*}
For the general relation between the entries of the inverse matrix and
quasideterminants, see \cite{GGRW}.

\begin{proposition}\label{prp:invariance_property}
Suppose that the displayed quasideterminants are defined, that the upper-left
$(n-1)\times(n-1)$ block of $X$ is invertible, and that the lower-right
$(n-1)\times(n-1)$ block of $Y$ is invertible. The corner quasideterminant
is invariant under left and right multiplication by matrices of the
following forms:
\begin{align}\label{XAY}
|A|_{n,1} = |X A Y|_{n,1}
\end{align}
where
\begin{align}
X = \begin{pmatrix}
x_{1,1} &\dots & x_{1,n-1}& 0\\
\vdots &\ddots &\vdots & \vdots \\
x_{n-1,1} &\dots& x_{n-1,n-1} & 0\\
x_{n,1} & \dots & x_{n,n-1} & 1
\end{pmatrix} \ \ ;\ \
Y = \begin{pmatrix}
1 & 0 &\dots & 0 \\
y_{2,1}&y_{2,2}&\dots & y_{2,n} \\
\vdots & \vdots &\ddots & \vdots \\
y_{n,1} & y_{n,2} & \dots & y_{n,n}
\end{pmatrix}
\end{align}
\end{proposition}
\begin{proof}
We write out the calculation for $n=3$; the same block cancellation proves
the statement for arbitrary $n$. For left multiplication, the product
$XA$ has the block form
\begin{align*}
XA
&=
\begin{pmatrix}
\begin{pmatrix}x_{1,1}&x_{1,2}\\x_{2,1}&x_{2,2}\end{pmatrix}
&\begin{pmatrix}0\\0\end{pmatrix}\\[1ex]
\begin{pmatrix}x_{3,1}&x_{3,2}\end{pmatrix}&1
\end{pmatrix}
\begin{pmatrix}
\begin{pmatrix}a_{1,1}\\a_{2,1}\end{pmatrix}
&\begin{pmatrix}a_{1,2}&a_{1,3}\\a_{2,2}&a_{2,3}\end{pmatrix}\\[1ex]
a_{3,1}&\begin{pmatrix}a_{3,2}&a_{3,3}\end{pmatrix}
\end{pmatrix} \\
&=
\begin{pmatrix}
\begin{pmatrix}x_{1,1}&x_{1,2}\\x_{2,1}&x_{2,2}\end{pmatrix}
\begin{pmatrix}a_{1,1}\\a_{2,1}\end{pmatrix}
&
\begin{pmatrix}x_{1,1}&x_{1,2}\\x_{2,1}&x_{2,2}\end{pmatrix}
\begin{pmatrix}a_{1,2}&a_{1,3}\\a_{2,2}&a_{2,3}\end{pmatrix}\\[2ex]
a_{3,1}+\begin{pmatrix}x_{3,1}&x_{3,2}\end{pmatrix}
\begin{pmatrix}a_{1,1}\\a_{2,1}\end{pmatrix}
&
\begin{pmatrix}a_{3,2}&a_{3,3}\end{pmatrix}
+\begin{pmatrix}x_{3,1}&x_{3,2}\end{pmatrix}
\begin{pmatrix}a_{1,2}&a_{1,3}\\a_{2,2}&a_{2,3}\end{pmatrix}
\end{pmatrix}.
\end{align*}
In the defining formula, the upper-left block of $X$ cancels between the
inverse of the deleted submatrix and the first-column vector. Substitution
into the definition of the corner quasideterminant therefore gives
\begin{align*}
|XA|_{3,1}
&=a_{3,1}+
\begin{pmatrix}
x_{3,1} & x_{3,2}
\end{pmatrix}
\begin{pmatrix}
a_{1,1}\\
a_{2,1}
\end{pmatrix}\\
&\quad-
\left(
\begin{pmatrix}
a_{3,2} & a_{3,3}
\end{pmatrix}
+
\begin{pmatrix}
x_{3,1} & x_{3,2}
\end{pmatrix}
\begin{pmatrix}
a_{1,2} & a_{1,3}\\
a_{2,2} & a_{2,3}
\end{pmatrix}
\right)
\begin{pmatrix}
a_{1,2} & a_{1,3}\\
a_{2,2} & a_{2,3}
\end{pmatrix}^{-1}
\begin{pmatrix}
a_{1,1}\\
a_{2,1}
\end{pmatrix}=|A|_{3,1}.
\end{align*}
Similarly, the product $AY$ has the form
\begin{align*}
AY
&=
\begin{pmatrix}
\begin{pmatrix}a_{1,1}\\a_{2,1}\end{pmatrix}
&\begin{pmatrix}a_{1,2}&a_{1,3}\\a_{2,2}&a_{2,3}\end{pmatrix}\\[1ex]
a_{3,1}&\begin{pmatrix}a_{3,2}&a_{3,3}\end{pmatrix}
\end{pmatrix}
\begin{pmatrix}
1&\begin{pmatrix}0&0\end{pmatrix}\\[1ex]
\begin{pmatrix}y_{2,1}\\y_{3,1}\end{pmatrix}
&\begin{pmatrix}y_{2,2}&y_{2,3}\\y_{3,2}&y_{3,3}\end{pmatrix}
\end{pmatrix}\\
&=
\begin{pmatrix}
\begin{pmatrix}a_{1,1}\\a_{2,1}\end{pmatrix}
+\begin{pmatrix}a_{1,2}&a_{1,3}\\a_{2,2}&a_{2,3}\end{pmatrix}
\begin{pmatrix}y_{2,1}\\y_{3,1}\end{pmatrix}
&
\begin{pmatrix}a_{1,2}&a_{1,3}\\a_{2,2}&a_{2,3}\end{pmatrix}
\begin{pmatrix}y_{2,2}&y_{2,3}\\y_{3,2}&y_{3,3}\end{pmatrix}\\[2ex]
a_{3,1}+\begin{pmatrix}a_{3,2}&a_{3,3}\end{pmatrix}
\begin{pmatrix}y_{2,1}\\y_{3,1}\end{pmatrix}
&
\begin{pmatrix}a_{3,2}&a_{3,3}\end{pmatrix}
\begin{pmatrix}y_{2,2}&y_{2,3}\\y_{3,2}&y_{3,3}\end{pmatrix}
\end{pmatrix}.
\end{align*}
Here the lower-right block of $Y$ cancels between the last-row vector and
the inverse of the deleted submatrix. Thus substitution into the defining
formula gives
\begin{align*}
|AY|_{3,1}
&=a_{3,1}+
\begin{pmatrix}
a_{3,2} & a_{3,3}
\end{pmatrix}
\begin{pmatrix}
y_{2,1}\\
y_{3,1}
\end{pmatrix}\\
&\quad-
\begin{pmatrix}
a_{3,2} & a_{3,3}
\end{pmatrix}
\begin{pmatrix}
a_{1,2} & a_{1,3}\\
a_{2,2} & a_{2,3}
\end{pmatrix}^{-1}
\left(
\begin{pmatrix}
a_{1,1}\\
a_{2,1}
\end{pmatrix}
+
\begin{pmatrix}
a_{1,2} & a_{1,3}\\
a_{2,2} & a_{2,3}
\end{pmatrix}
\begin{pmatrix}
y_{2,1}\\
y_{3,1}
\end{pmatrix}
\right)=|A|_{3,1}.
\end{align*}
\end{proof}

We next record explicit expressions for quasideterminants of two special
classes of matrices.

\begin{definition}

For a matrix $A=(a_{a,b})_{a,b=1}^n$ with entries in a unital associative
algebra, define
\begin{equation}\label{eq:quantum_row_determinant}
\operatorname{rdet}A
=\sum_{\sigma\in S_n}\operatorname{sgn}(\sigma)
 a_{n,\sigma(n)}\cdots a_{1,\sigma(1)}.
\end{equation}
\end{definition}
For the following special matrices, the corner quasideterminant can be
computed explicitly.
\begin{proposition}\label{prp:nice_formulas_for_quasidet}
\begin{enumerate}
\item For a block matrix whose first upper diagonal consists of identity
blocks and whose blocks above the first upper diagonal are zero, one has
\begin{align}
\label{quarow}
\left|
\begin{pmatrix} a_{1,1} & 1&
0&\ldots& 0\\
a_{2,1}&a_{2,2}& 1&\ldots& 0\\
a_{3,1}&a_{3,2}&a_{3,3}&\ldots& 0\\
\vdots&\vdots&\vdots& \ddots& \vdots\\
a_{n,1}&a_{n,2}&a_{n,3}&\ldots&a_{n,n}
 \end{pmatrix}
\right|_{n,1} = (-1)^{1-n}
\operatorname{rdet}\begin{pmatrix} a_{1,1} & 1&
0&\ldots& 0\\
a_{2,1}&a_{2,2}& 1&\ldots& 0\\
a_{3,1}&a_{3,2}&a_{3,3}&\ldots& 0\\
\vdots&\vdots&\vdots& \ddots& \vdots\\
a_{n,1}&a_{n,2}&a_{n,3}&\ldots&a_{n,n}
 \end{pmatrix}
\end{align}
\item If $b_i\in\mathcal A^\times$ for $1\leq i<n$, then

\begin{align}
    \left|
\begin{pmatrix} a_{1} & b_1&
0&\ldots& 0\\
0&a_{2}& b_2&\ldots& 0\\
0&0&a_{3}&\ddots& 0\\
\vdots&\vdots&\ddots& \ddots& b_{n-1}\\
0&0&\ldots&0&a_{n}
 \end{pmatrix}\right|_{n,1}=(-1)^{1-n}
a_nb_{n-1}^{-1}a_{n-1}b_{n-2}^{-1}\cdots a_2b_1^{-1}a_1.
\end{align}
\end{enumerate}
\end{proposition}
\begin{proof}
For the first assertion, the general identity is standard; see
\cite{GGRW}. For illustration, let $A$ be the displayed block matrix with
$n=3$, identity blocks on the first upper diagonal, and zero blocks above
that diagonal. Then the defining formula gives
\begin{align*}
\left|A
\right|_{3,1} &=
a_{3,1}-\begin{pmatrix}a_{3,2}&a_{3,3}\end{pmatrix}
\begin{pmatrix}1 & 0\\
a_{2,2}&1
\end{pmatrix}^{-1}
\begin{pmatrix}a_{1,1}\\
a_{2,1}\end{pmatrix}  =
a_{3,1} +a_{3,3}a_{2,2}a_{1,1} -
a_{3,2}a_{1,1}-a_{3,3}a_{2,1}.
\end{align*}
This expression coincides with the row determinant
\begin{align*}
\operatorname{rdet}\begin{pmatrix} a_{1,1} &1 & 0\\
a_{2,1} & a_{2,2}& 1\\
a_{3,1} & a_{3,2} & a_{3,3}
\end{pmatrix} =
a_{3,3}a_{2,2}a_{1,1} +
a_{3,1} -
 a_{3,2}a_{1,1} -a_{3,3}a_{2,1}.
\end{align*}
The general statement follows by the same expansion.

For the second assertion, use Proposition~\ref{prp:invariance_property}
and the factorization
\begin{align*}
\begin{pmatrix} a_{1} & b_1&
0&\ldots& 0\\
0&a_{2}& b_2&\ldots& 0\\
0&0&a_{3}&\ddots& 0\\
\vdots&\vdots&\ddots& \ddots& b_{n-1}\\
0&0&\ldots&0&a_{n}
 \end{pmatrix}=\begin{pmatrix} b_{1} & 0&
0&\ldots& 0\\
0&b_{2}& 0&\ldots& 0\\
0&0&b_{3}&\ddots& 0\\
\vdots&\vdots&\ddots& \ddots& 0\\
0&0&\ldots&0&1
 \end{pmatrix}\begin{pmatrix}  b_1^{-1}a_{1}  &1&
0&\ldots& 0\\
0&b_2^{-1}a_{2}& 1&\ldots& 0\\
0&0&b_3^{-1}a_{3}&\ddots& 0\\
\vdots&\vdots&\ddots& \ddots&1\\
0&0&\ldots&0&a_{n}
 \end{pmatrix}.
\end{align*}
The left factor does not change the corner quasideterminant. Applying the
first assertion to the right factor and evaluating its row determinant
gives the stated product.
\end{proof}

\subsection{Commutation relations}
We now regard $L(u)$ as an $n\times n$ matrix whose entries are
$p\times p$ matrices over $\mathcal U(\mathfrak{gl}_N)((u^{-1}))$.
In this subsection only, it is convenient to retain the notation
\begin{equation*}
\mathbf L_{a,b}(u)=u\delta_{a,b}+\mathbf E_{a,b},
\qquad
L_{a,b}^{i,j}(u)=(\mathbf L_{a,b}(u))_{i,j}.
\end{equation*}
\begin{definition}
Suppose that the block submatrix obtained by deleting the last block row
and the first block column is invertible in a suitable overring. The corner
quasideterminant of the block matrix $L(u)$ is
\begin{equation}\label{eq:corner_quasideterminant}
\begin{aligned}
T(u)=|L(u)|_{n,1}
&=\mathbf L_{n,1}(u)
-\begin{pmatrix}
\mathbf L_{n,2}(u)&\cdots&\mathbf L_{n,n}(u)
\end{pmatrix}\\
&\quad\times
\begin{pmatrix}
\mathbf L_{1,2}(u)&\cdots&\mathbf L_{1,n}(u)\\
\vdots&\ddots&\vdots\\
\mathbf L_{n-1,2}(u)&\cdots&\mathbf L_{n-1,n}(u)
\end{pmatrix}^{-1}
\begin{pmatrix}
\mathbf L_{1,1}(u)\\
\vdots\\
\mathbf L_{n-1,1}(u)
\end{pmatrix}.
\end{aligned}
\end{equation}
\end{definition}

The matrix $L(u)=u+E$ is invertible as a formal Laurent series:
$$
L(u)^{-1}
=u^{-1}\sum_{r\geq0}(-u^{-1}E)^r
\in\operatorname{Mat}_N
\bigl(\mathcal U(\mathfrak{gl}_N)((u^{-1}))\bigr).
$$
Using the same $n\times n$ block decomposition as in
\eqref{eq:block_L_operator}, write
\begin{equation}\label{eq:full_inverse_block_matrix}
L(u)^{-1}
=\begin{pmatrix}
\mathbf M_{1,1}(u)&\mathbf M_{1,2}(u)&\cdots&\mathbf M_{1,n}(u)\\
\mathbf M_{2,1}(u)&\mathbf M_{2,2}(u)&\cdots&\mathbf M_{2,n}(u)\\
\vdots&\vdots&\ddots&\vdots\\
\mathbf M_{n,1}(u)&\mathbf M_{n,2}(u)&\cdots&\mathbf M_{n,n}(u)
\end{pmatrix},
\end{equation}
where each $p\times p$ block is written explicitly as
\begin{equation*}
\mathbf M_{a,b}(u)=\sum_{i,j=1}^p e_{i,j}\otimes
M_{a,b}^{i,j}(u).
\end{equation*}
Thus
\begin{equation}\label{eq:full_inverse_relations}
\sum_{b=1}^n\sum_{j=1}^p
L_{a,b}^{i,j}(u)M_{b,c}^{j,k}(u)
=\delta_{a,c}\delta_{i,k},
\qquad
\sum_{b=1}^n\sum_{j=1}^p
M_{a,b}^{i,j}(u)L_{b,c}^{j,k}(u)
=\delta_{a,c}\delta_{i,k}.
\end{equation}
\paragraph{Example.} When $n=2$, the block matrices are
\begin{equation*}
L(u)=
\begin{pmatrix}
\mathbf L_{1,1}(u)&\mathbf L_{1,2}(u)\\
\mathbf L_{2,1}(u)&\mathbf L_{2,2}(u)
\end{pmatrix},
\qquad
L(u)^{-1}=
\begin{pmatrix}
\mathbf M_{1,1}(u)&\mathbf M_{1,2}(u)\\
\mathbf M_{2,1}(u)&\mathbf M_{2,2}(u)
\end{pmatrix}.
\end{equation*}
The relations $L(u)L(u)^{-1}=L(u)^{-1}L(u)=1$ become
\begin{align*}
\begin{pmatrix}
\mathbf L_{1,1}(u)&\mathbf L_{1,2}(u)\\
\mathbf L_{2,1}(u)&\mathbf L_{2,2}(u)
\end{pmatrix}\begin{pmatrix}
\mathbf M_{1,1}(u)&\mathbf M_{1,2}(u)\\
\mathbf M_{2,1}(u)&\mathbf M_{2,2}(u)
\end{pmatrix} =
\begin{pmatrix}
\mathbf M_{1,1}(u)&\mathbf M_{1,2}(u)\\
\mathbf M_{2,1}(u)&\mathbf M_{2,2}(u)
\end{pmatrix}\begin{pmatrix}
\mathbf L_{1,1}(u)&\mathbf L_{1,2}(u)\\
\mathbf L_{2,1}(u)&\mathbf L_{2,2}(u)
\end{pmatrix}
=\begin{pmatrix}\mathbf1_p&0\\0&\mathbf1_p\end{pmatrix}.
\end{align*}

Whenever the corner quasideterminant $T(u)=|L(u)|_{n,1}$ is defined in a
localized ring $\mathcal A$ containing
$\mathcal{U}(\mathfrak{gl}_N)((u^{-1}))$,
Proposition~\ref{prp:quasidet_and_inverse} gives
\begin{equation}\label{eq:inverse_corner_block}
\mathbf M_{1,n}(u)=T(u)^{-1}.
\end{equation}
Recall that the commutation relations among the entries of $L(u)$ can be
written in RTT form \eqref{eq:yangian_relation} as
\begin{align}\label{Y}
(u-v)L_{a,b}^{k,i}(u)L_{c,d}^{s,j}(v)
+L_{c,b}^{s,i}(u)L_{a,d}^{k,j}(v) =(u-v)L_{c,d}^{s,j}(v)L_{a,b}^{k,i}(u)
+L_{c,b}^{s,i}(v)L_{a,d}^{k,j}(u).
\end{align}
We use this identity to establish the two properties required for
Hamiltonian reduction.

\begin{proposition}\label{prp:inverse_block_commutation}
For $1\leq a\leq n-1$ and $2\leq b\leq n$, every entry of
$\mathbf L_{a,b}(u)$ commutes with every entry of
$\mathbf M_{1,n}(v)$:
\begin{equation}\label{eq:inverse_block_commutation}
[L_{a,b}^{i,j}(u),\mathbf{M}_{1,n}(v)]=0,
\end{equation}
where the commutator is understood entrywise.
Consequently, whenever $T(v)$ is defined and invertible in a common
overring,
\begin{equation}\label{LT}
[L_{a,b}^{i,j}(u),T(v)]=0.
\end{equation}
\end{proposition}

\begin{proof}
Multiply \eqref{Y} by the entries of $L(v)^{-1}$ and use
\eqref{eq:full_inverse_relations}. The resulting identity is
\begin{align}\label{ML}
&(u-v)M_{e,f}^{q,m}(v)L_{a,b}^{k,i}(u)
+\sum_{c=1}^n\sum_{s=1}^pM_{e,c}^{q,s}(v)L_{c,b}^{s,i}(u)
\delta_{k,m}\delta_{a,f}\notag\\
&\qquad=(u-v)L_{a,b}^{k,i}(u)M_{e,f}^{q,m}(v)
+\sum_{d=1}^n\sum_{j=1}^p\delta_{q,i}\delta_{e,b}
L_{a,d}^{k,j}(u)M_{d,f}^{j,m}(v).
\end{align} Set $e=1$ and $f=n$. If $a\leq n-1$, then
$\delta_{a,f}=0$; if $b\geq2$, then $\delta_{e,b}=0$. Hence both terms
containing Kronecker symbols vanish, and the relation above reduces to
$$
M_{1,n}^{q,m}(v)L_{a,b}^{k,i}(u)
=L_{a,b}^{k,i}(u)M_{1,n}^{q,m}(v).
$$
Finally,
\eqref{eq:inverse_corner_block} and the general rule
$$
[A,B^{-1}]=-B^{-1}[A,B]B^{-1}
$$
give \eqref{LT}.

\end{proof}

\begin{proposition}\label{prp:inverse_block_opposite_RTT}
The entries of the block $\mathbf M_{1,n}(u)$ satisfy the opposite RTT
relation
\begin{align}\label{M1n}
&(u-v)M_{1,n}^{i,j}(u)M_{1,n}^{q,m}(v)
+M_{1,n}^{i,m}(u)M_{1,n}^{q,j}(v)\notag\\
&\qquad=(u-v)M_{1,n}^{q,m}(v)M_{1,n}^{i,j}(u)
+M_{1,n}^{i,m}(v)M_{1,n}^{q,j}(u).
\end{align}
Equivalently, in matrix form,
\begin{equation}\label{eq:inverse_block_opposite_RTT}
\bigl(\mathbf M_{1,n}(u)\bigr)_1
\bigl(\mathbf M_{1,n}(v)\bigr)_2R(u-v)
=R(u-v)\bigl(\mathbf M_{1,n}(v)\bigr)_2
\bigl(\mathbf M_{1,n}(u)\bigr)_1.
\end{equation}
\end{proposition}

\begin{proof}
Multiply \eqref{ML} by the entries of $L(u)^{-1}$ and use
\eqref{eq:full_inverse_relations} once more. This gives
\begin{align}\label{MM}
&(u-v)M_{c,h}^{i,j}(u)M_{e,f}^{q,m}(v)
+M_{c,f}^{i,m}(u)M_{e,h}^{q,j}(v)\notag\\
&\qquad=(u-v)M_{e,f}^{q,m}(v)M_{c,h}^{i,j}(u)
+M_{c,f}^{i,m}(v)M_{e,h}^{q,j}(u).
\end{align}
Taking $c=e=1$ and $h=f=n$ gives \eqref{M1n}. Relations of this type are
equivalent to the matrix RTT relation with the rational $R$-matrix.

Explicit computations are written out in Appendix~\ref{app:det}.
\end{proof}

\begin{corollary}\label{cor:corner_quasideterminant_RTT}
Whenever $T(u)$ is defined and invertible in a common overring, it satisfies
the RTT relation
\begin{equation}\label{TT}
R(u-v)T_1(u)T_2(v)=T_2(v)T_1(u)R(u-v).
\end{equation}
\end{corollary}

\begin{proof}
By \eqref{eq:inverse_corner_block}, the opposite RTT relation
\eqref{eq:inverse_block_opposite_RTT} is the corresponding relation for
$T(u)^{-1}$. Taking inverses and multiplying by $R(u-v)$ on the left and
right gives \eqref{TT}.
\end{proof}

\begin{remark}
The formal inverse $L(u)^{-1}$ is sufficient for the two propositions
above. The reduction argument in Section~\ref{sec:localization_independence}
shows why the result is independent of the choice of compatible
localization.
\end{remark}
\section{Reduction and W-algebras}
\label{sec:W_algebras}
In the first subsection, we recall the basic definitions and results of
Brundan--Kleshchev \cite{BK} on finite $W$-algebras, building on the earlier
rectangular construction of Ragoucy--Sorba \cite{RS}. These works provide
explicit formulas for the invariants, although the origin of these formulas,
especially that of the shifts, remains unclear.

\subsection{The finite W-algebra}

Define the following subalgebras of
$\mathfrak{gl}_N$ directly in terms of the matrix units:
\begin{equation}\label{eq:quantum_reduction_subalgebras}
\begin{aligned}
\mathfrak m_n
&=\operatorname{span}\bigl\{E_{a,b}^{i,j}\mid a>b\bigr\},
&
\mathfrak m_n^+
&=\operatorname{span}\bigl\{E_{a,b}^{i,j}\mid a<b\bigr\},\\
\mathfrak l_n
&=\operatorname{span}\bigl\{E_{a,a}^{i,j}\mid 1\leq a\leq n\bigr\},
&
\mathfrak p_n
&=\operatorname{span}\bigl\{E_{a,b}^{i,j}\mid a\leq b\bigr\}
=\mathfrak l_n\oplus\mathfrak m_n^+.
\end{aligned}
\end{equation}
Here and below, the internal indices $i,j$ range from $1$ to $p$ in each
displayed span. Thus $\mathfrak m_n$ and $\mathfrak m_n^+$ are the strictly
lower- and upper-block-triangular subalgebras, $\mathfrak l_n$ is the
block-diagonal subalgebra, and $\mathfrak p_n$ is the upper-block-triangular
parabolic subalgebra.

Define a character $\psi:\mathfrak m_n\to\mathbb C$ by
\begin{equation}\label{eq:quantum_character}
\psi(E_{a,b}^{i,j})=\delta_{i,j}\delta_{a,b+1},
\qquad a>b.
\end{equation}
Let $I_\psi=\mathcal U(\mathfrak{gl}_N)(\mathfrak m_n-\psi)$ be the
corresponding left ideal.
Set
$$
Q_\psi=\mathcal U(\mathfrak{gl}_N)/I_\psi,
\qquad \overline1=1+I_\psi.
$$

\begin{definition}
Define the  rectangular finite $W$-algebra is
\begin{equation}\label{eq:finite_W_algebra}
\mathcal W_n(\mathfrak{gl}_p)
=Q_\psi^{\mathfrak m_n},
\end{equation}
where $\mathfrak m_n$ acts on $Q_\psi$ by the adjoint action. Explicitly,
$$
Q_\psi^{\mathfrak m_n}
=\{f\overline1\in Q_\psi\mid [x,f]\overline1=0
\text{ for every }x\in\mathfrak m_n\}.
$$
\end{definition}
The PBW decomposition associated with
$\mathfrak{gl}_N=\mathfrak p_n\oplus\mathfrak m_n$ gives a vector-space
isomorphism
$$
\mathcal U(\mathfrak p_n)\otimes\mathcal U(\mathfrak m_n)
\xrightarrow{\sim}\mathcal U(\mathfrak{gl}_N).
$$
Applying the character $\psi$ to the second tensor factor defines the
projection
$$
\operatorname{pr}_\psi:
\mathcal U(\mathfrak{gl}_N)\longrightarrow\mathcal U(\mathfrak p_n).
$$
Its kernel is $I_\psi$, so it induces an isomorphism
$Q_\psi\simeq\mathcal U(\mathfrak p_n)$ of left
$\mathcal U(\mathfrak p_n)$-modules.
Under this identification, the finite $W$-algebra is realized as a
subalgebra of $\mathcal U(\mathfrak p_n)$, giving an injective algebra
homomorphism
$$
i:\mathcal W_n(\mathfrak{gl}_p)
\hookrightarrow\mathcal U(\mathfrak p_n).
$$
The assignments $E_{a,a}^{i,j}\mapsto E^{i,j}_{[a]}$  induce a canonical algebra
isomorphism
\begin{equation}\label{eq:Levi_tensor_identification}
\kappa:\mathcal U(\mathfrak l_n)
\xrightarrow{\sim}\mathcal U(\mathfrak{gl}_p)^{\otimes n}.
\end{equation}
We use different notation on the two sides of this isomorphism. Set
$$
\mathbf E_a=\mathbf E_{a,a}
=\sum_{i,j=1}^p e_{i,j}\otimes E_{a,a}^{j,i}
\in\operatorname{Mat}_p(\mathbb C)\otimes\mathcal U(\mathfrak l_n)
$$
and
$$
 E_{[a]}
=\sum_{i,j=1}^p e_{i,j}\otimes E^{j,i}_{[a]}
\in\operatorname{Mat}_p(\mathbb C)
\otimes\mathcal U(\mathfrak{gl}_p)^{\otimes n}.
$$
Thus $\kappa(\mathbf E_a)=E_{[a]}$, where unbracketed block labels refer to
diagonal blocks in $\mathfrak{gl}_N$ and square-bracketed subscripts
label tensor factors.
\subsection{Brundan--Kleshchev generating matrix}
We now construct the corresponding generating matrix in the finite
$W$-algebra, using the row determinant defined in
\eqref{eq:quantum_row_determinant}.
In the quotient by $I_\psi$, the strictly upper-block entries of
\eqref{eq:block_L_operator} satisfy
$$
\mathbf E_{a,a+1}=\mathbf1_p,
\qquad
\mathbf E_{a,b}=0\quad(b>a+1).
$$
We consider
\begin{equation}
    \label{eq:reduced_L_op}L^{\mathrm{red}}(u)=\begin{pmatrix}
   u-p(n-1)+ \mathbf E_{1,1}&\mathbf1_p&\cdots&0\\
\mathbf E_{2,1}&u-p(n-2)+\mathbf E_{2,2}&\ddots&\vdots\\
\vdots&\ddots&\ddots&\mathbf1_p\\
\mathbf E_{n,1}&\cdots&\mathbf E_{n,n-1}&u+\mathbf E_{n,n}
\end{pmatrix}.\end{equation}
This is a matrix with entries in $Q_\psi[u]$.

\begin{proposition}[Brundan--Kleshchev \cite{BK}]
The matrix entries of the polynomial
\begin{equation} \label{eq:quantum_generating_function}
\mathbf{T}(u) = \operatorname{rdet}\bigl(L^{\mathrm{red}}(u) \bigr)
\end{equation}
belong to $\mathcal{W}_n(\mathfrak{gl}_p)[u]$, and their coefficients
generate $\mathcal{W}_n(\mathfrak{gl}_p)$.
\end{proposition}
\paragraph{Example.} In case $N=2p, ~n=2$ the Brundan-Kleschev matrix has form:
$$\mathbf{T}(u)=(u+\mathbf{E}_{2,2})(u-p+\mathbf{E}_{1,1})-\mathbf{E}_{2,1}.$$
If we set $\mathbf{E}_{2,1}=0$, than we recover the usual truncated Yangian generating matrix.

The generalization to generic $N$ is obvious. Consider the projection onto the block-diagonal factor. It gives an algebra homomorphism
$$
\eta:\mathcal U(\mathfrak p_n)\longrightarrow\mathcal U(\mathfrak l_n),
\qquad
\mathcal U(\mathfrak p_n)/\mathfrak m_n^+\mathcal U(\mathfrak p_n)
\simeq\mathcal U(\mathfrak l_n).
$$
Under the composition $\eta\circ i$, the image of the generating matrix is
\begin{equation*}
(\eta\circ i)\bigl(\mathbf{T}(u)\bigr)
=(u+\mathbf E_n)(u-p+\mathbf E_{n-1})\cdots
(u-(n-1)p+\mathbf E_1).
\end{equation*}
To remove these scalar shifts, we apply the automorphism
$
E_{a,a}^{i,j}\mapsto
E_{a,a}^{i,j}+(n-a)p\delta_{i,j}.$
We then use $\kappa$ to identify the enveloping algebra of the
block-diagonal subalgebra with $\mathcal U(\mathfrak{gl}_p)^{\otimes n}$.
Composing these maps defines the quantum Miura homomorphism
\begin{equation*}
\Theta:
\mathcal W_n(\mathfrak{gl}_p)
\longrightarrow\mathcal U(\mathfrak{gl}_p)^{\otimes n},
\end{equation*}
which acts on the generating matrix as follows:
\begin{align}\label{eq:Miura_map}
\Theta(\mathbf{T}(u))
=(u+E_{[n]})(u+ E_{[n-1]})\cdots(u+ E_{[1]}).
\end{align}This is the coproduct-evaluation
generating matrix associated with
\eqref{eq:quantum_yangian_generating_matrix}. The quantum Miura
homomorphism is summarized by the diagram
\begin{equation}\label{eq:quantum_Miura_diagram}
\begin{tikzcd}
{\mathcal W_n(\mathfrak{gl}_p)}
    \arrow[r,hook,"i"]
    \arrow[d,"\Theta"']
&{\mathcal U(\mathfrak p_n)}
    \arrow[d]\\
{Y^{(n)}(\mathfrak{gl}_p)}
    \arrow[r,hook]
&{\mathcal U(\mathfrak{gl}_p)}^{\otimes n}.
\end{tikzcd}
\end{equation}

\begin{theorem}[Brundan--Kleshchev \cite{BK}] \label{thm:truncated_yangian_w_alg}
The homomorphism $\Theta$ identifies $\mathcal W_n(\mathfrak{gl}_p)$ with
the level-$n$ truncated Yangian $Y^{(n)}(\mathfrak{gl}_p)$ generated by the
coefficients of \eqref{eq:quantum_yangian_generating_matrix}.
\end{theorem}
\subsection{Reduction of the quasideterminant}
\label{sec:localization_independence}
Assume that we have constructed a ring $\mathcal{A}$ in which the
quasideterminant is defined, or in which its rescaled version is defined,
and consider any such compatible localization. We now give a general
argument showing that the reduction does not depend on the particular
choice of localization.
Denote by $\mathcal{I}_\psi$ the left ideal generated by the reduction
constraints $m-\psi(m)$, for $m\in\mathfrak{m}_n$. The quotient map restricts to the homomorphism of algebras of invariants $\mathcal{A}^{\mathfrak{m}_n}\rightarrow (\mathcal{A}/\mathcal{I}_\psi)^{\mathfrak{m}_n}.$ 

As shown above, the matrix entries of the quasideterminant belong to
$\mathcal{A}^{\mathfrak{m}_n}$. Thus, it follows that the reduction
$T(u)\overline{1}$ remains invariant and satisfies the RTT relation. Let us
identify the reduced quasideterminant with the Brundan--Kleshchev generating matrix.
\begin{theorem}\label{th:quasidet_and_BK}
The reduction of the quasideterminant coincides with the 
Brundan--Kleshchev generating matrix:
\begin{equation}\label{eq:quasidet_BK_polynomial}
(-1)^{n-1}T(u)\overline1=(-1)^{n-1}|L^{\mathrm{red}}(u)|=\mathbf T(u).
\end{equation}
\end{theorem}

To illustrate the proof strategy, consider the case $n=3$. Under reduction, the matrix entries satisfy $L_{1,2}^{i,j}(u)\overline{1}=L_{2,3}^{i,j}(u)\overline{1}=\delta_{i,j}$ and $L_{1,3}^{i,j}(u)\overline{1}=0.$
Let $$\mathbf{q}(u)=
\begin{pmatrix}
\mathbf L_{1,2}(u)&\mathbf L_{1,3}(u)\\
\mathbf L_{2,2}(u)&\mathbf L_{2,3}(u)
\end{pmatrix}^{-1}\begin{pmatrix}
\mathbf L_{1,1}(u)\\
\mathbf L_{2,1}(u)
\end{pmatrix}=\begin{pmatrix}
    \mathbf{q}_1(u)\\\mathbf{q}_2(u)
\end{pmatrix},$$
and we formally set $\mathbf{q}_0(u)=0.$ The quasideterminant is then given by $$|L(u)|_{3,1}=\mathbf{L}_{3,1}(u)-\begin{pmatrix}
    \mathbf{L}_{3,3}(u)&\mathbf{L}_{3,2}(u)
\end{pmatrix} \begin{pmatrix}
    \mathbf{q}_1(u)\\\mathbf{q}_2(u)
\end{pmatrix}.$$
We denote the entries of $\mathbf{q}_b(u)$ by $q_{b}^{j,k}(u)$.
The following commutator identities hold: $$\sum_{j=1}^p[L_{1,2}^{i,j}(u),q_1^{j,k}(u)]=p\delta_{i,k},\quad \sum_{j=1}^p[L_{1,3}^{i,j}(u),q_2^{j,k}(u)]=p\delta_{i,k}, $$
and $$\sum_{j=1}^p[L_{2,3}^{i,j}(u),q_2^{j,k}(u)]=-pq_1^{i,k}(u).$$
We can combine these into a single formula: $$\sum_{j=1}^p
\bigl[L_{a,b+1}^{i,j}(u),q_b^{j,k}(u)\bigr]
=-pq_{a-1}^{i,k}(u)+p\delta_{a,1}\delta_{i,k}.$$
Next, consider the matrix identity $$\begin{pmatrix}
\mathbf L_{1,2}(u)&\mathbf L_{1,3}(u)\\
\mathbf L_{2,2}(u)&\mathbf L_{2,3}(u)
\end{pmatrix}\begin{pmatrix}
    \mathbf{q}_1(u)\\\mathbf{q}_2(u)
\end{pmatrix}=\begin{pmatrix}
\mathbf L_{1,1}(u)\\
\mathbf L_{2,1}(u)
\end{pmatrix}.$$
Writing this identity in components and reordering factors as needed, we obtain the following equivalent relations:\begin{align*}2p\delta_{i,k}+\sum_{j=1}^p q_{1}^{j,k}(u)L_{1,2}^{i,j}(u)+q_{2}^{j,k}(u)L_{1,3}^{i,j}(u)=L_{1,1}^{i,k}(u),\\ -pq_1^{i,k}+ \sum_{j=1}^p L_{2,2}^{i,j}(u)q_{1}^{j,k}(u)+q_{2}^{j,k}(u)L_{2,3}^{i,j}(u)=L_{2,1}^{i,k}(u).
\end{align*}
Applying the quotient map gives:
\begin{align*} q_{1}^{i,k}(u)\overline{1}=L_{1,1}^{i,k}(u)-2p\delta_{i,k},\\q_{2}^{i,k}(u)\overline{1} +\sum_{j=1}^p (L_{2,2}^{i,j}(u)-p\delta_{i,j})(u)q_{1}^{j,k}(u)\overline{1}=L_{2,1}^{i,k}(u).
\end{align*}
Solving for $\mathbf{q}_1(u)$ and $\mathbf{q}_2(u)$ gives: $$ \begin{pmatrix}
    \mathbf{q}_1(u)\overline{1}\\\mathbf{q}_2(u)\overline{1}
\end{pmatrix}=\begin{pmatrix}
\mathbf 1_p&0\\
\mathbf L_{2,2}(u)-p&\mathbf 1_p
\end{pmatrix}^{-1}\begin{pmatrix}
\mathbf L_{1,1}(u)\\
\mathbf L_{2,1}(u)
\end{pmatrix}.$$
Substituting this expression into the quasideterminant formula yields $$|L(u)|_{3,1}\overline{1}=\left|\begin{pmatrix}\mathbf{L}_{1,1}(u)-2p&\mathbf 1_p&0\\\mathbf{L}_{2,1}(u)&
\mathbf L_{2,2}(u)-p&\mathbf 1_p\\
\mathbf{L}_{3,1}(u)&\mathbf{L}_{3,2}(u)&\mathbf{L}_{3,3}(u)
\end{pmatrix}\right|_{3,1}=\mathrm{rdet}\begin{pmatrix}\mathbf{L}_{1,1}(u)-2p&\mathbf 1_p&0\\\mathbf{L}_{2,1}(u)&
\mathbf L_{2,2}(u)-p&\mathbf 1_p\\
\mathbf{L}_{3,1}(u)&\mathbf{L}_{3,2}(u)&\mathbf{L}_{3,3}(u)
\end{pmatrix}.$$
The final equality identifies this reduced quasideterminant with the Brundan--Kleshchev generating matrix.
\begin{proof}
Set
$$
\mathbf A(u)=
\begin{pmatrix}
\mathbf L_{1,2}(u)&\cdots&\mathbf L_{1,n}(u)\\
\vdots&\ddots&\vdots\\
\mathbf L_{n-1,2}(u)&\cdots&\mathbf L_{n-1,n}(u)
\end{pmatrix},
\qquad
\mathbf c(u)=
\begin{pmatrix}
\mathbf L_{1,1}(u)\\
\vdots\\
\mathbf L_{n-1,1}(u)
\end{pmatrix}.
$$
Write
$
\mathbf G_{a,b}(u)$ for blocks of 
$\mathbf A(u)^{-1}$.

Instead of reducing the whole inverse matrix, we directly reduce the column:
\begin{equation}\label{eq:quasidet_auxiliary_column}
\mathbf q(u)
=\mathbf A(u)^{-1}\mathbf c(u),
\qquad
\mathbf q(u)=
\begin{pmatrix}\mathbf q_1(u)\\ \vdots\\ \mathbf q_{n-1}(u)
\end{pmatrix}.
\end{equation}
Thus
\begin{equation}\label{eq:quasidet_via_auxiliary_column}
T(u)\overline1
=\left(
\mathbf L_{n,1}(u)
-\sum_{b=1}^{n-1}\mathbf L_{n,b+1}(u)\mathbf q_b(u)
\right)\overline1.
\end{equation}
Let us establish the reduction step by step.
\begin{itemize}
    \item \textbf{Commutation relations}. Put $\mathbf q_0(u)=0$. For
$1\leq a\leq n-1$ and $a\leq b\leq n-1$, one has
\begin{equation}\label{eq:key_q_commutator}
\sum_{j=1}^p
\bigl[L_{a,b+1}^{i,j}(u),q_b^{j,k}(u)\bigr]
=-pq_{a-1}^{i,k}(u)+p\delta_{a,1}\delta_{i,k}.
\end{equation}
To verify this identity, write
$$
q_b^{j,k}(u)
=\sum_{d=1}^{n-1}\sum_{t=1}^p
G_{b,d}^{j,t}(u)L_{d,1}^{t,k}(u).
$$
In components, the inverse-matrix commutator formula reads
\begin{align*}
\bigl[L_{a,b+1}^{i,j}(u),G_{r,d}^{s,t}(u)\bigr]
&=\delta_{a,d}\delta_{t,i}G_{r,b+1}^{s,j}(u)
-\delta_{r,b}\delta_{s,j}G_{a-1,d}^{i,t}(u),
\end{align*}
with the convention that blocks of $\mathbf G$ outside the range
$1,\dots,n-1$ are zero. Its contribution to
\eqref{eq:key_q_commutator}, after summing over $j,d,t$, is
$$
\operatorname{tr}\bigl(\mathbf G_{b,b+1}(u)\bigr)
L_{a,1}^{i,k}(u)-pq_{a-1}^{i,k}(u),
$$
where the first term is absent when $b=n-1$. On the other hand,
\eqref{eq:generator_commutation_relations} gives
\begin{align*}
\bigl[L_{a,b+1}^{i,j}(u),L_{d,1}^{t,k}(u)\bigr]
&=\delta_{a,1}\delta_{i,k}L_{d,b+1}^{t,j}(u)
-\delta_{d,b+1}\delta_{t,j}L_{a,1}^{i,k}(u).
\end{align*}
Using the inverse relation
$\sum_d G_{b,d}^{i,k}(u) L_{d,b+1}^{k,j}(u)=\delta^{i,j}$, the contribution
from this commutator is
$$
p\delta_{a,1}\delta_{i,k}
-\operatorname{tr}\bigl(\mathbf G_{b,b+1}(u)\bigr)
L_{a,1}^{i,k}(u).
$$
The trace terms cancel, proving \eqref{eq:key_q_commutator}. 
\item \textbf{Recurrence}.
The defining equation
$\mathbf A(u)\mathbf q(u)=\mathbf c(u)$ reads componentwise
\begin{equation}\label{eq:q_defining_equation}
\sum_{b=1}^{n-1}\sum_{j=1}^p
L_{a,b+1}^{i,j}(u)q_b^{j,k}(u)
=L_{a,1}^{i,k}(u).
\end{equation}
For $b\geq a$, move $L_{a,b+1}^{i,j}$ to the right before applying the
result to $\overline1$.  The reduction relations give
$$
L_{a,b+1}^{i,j}(u)\overline1
=\delta_{a,b}\delta_{i,j}\overline1,
\qquad b\geq a,
$$
while \eqref{eq:key_q_commutator} gives the correction term. We now split
the sum and commute the generators:
\begin{align*}
    \sum_{b=1}^{n-1}\sum_{j=1}^p
L_{a,b+1}^{i,j}(u)q_b^{j,k}(u)\overline{1}=\sum_{j=1}^p\left[\sum_{b= a}^{n-1}
L_{a,b+1}^{i,j}(u)q_b^{j,k}(u)+ \sum_{b=1}^{a-2}
L_{a,b+1}^{i,j}(u)q_b^{j,k}(u) +L_{a,a}^{i,j}(u)q_{a-1}^{j,k}(u)\right]\overline{1}\\
=\sum_{j=1}^p\sum_{b= a}^{n-1}
\delta_{a,b}\delta^{i,j}q_b^{j,k}(u)\overline{1}+ \sum_{j=1}^p\bigg(\sum_{b=1}^{a-2}
L_{a,b+1}^{i,j}(u)q_b^{j,k}(u) +L_{a,a}^{i,j}(u)q_{a-1}^{j,k}(u)\bigg)\overline{1}- \sum_{b=a}^{n-1}p(q_{a-1}^{i,k}-\delta_{a,1}\delta^{i,k})\overline{1}.\end{align*}
Next, we sum over $j$ and $b$ in the first and third terms. Thus the
defining relation \eqref{eq:q_defining_equation} becomes
\begin{align*} 
q_a^{i,k}\overline{1}+\sum_{j=1}^p\bigg(\sum_{b=1}^{a-2}
L_{a,b+1}^{i,j}(u)q_b^{j,k}(u)
+\bigl(L_{a,a}^{i,j}(u)-p(n-a)\delta^{i,j}\bigr)
q_{a-1}^{j,k}(u)\bigg)\overline{1}
+p(n-1)\delta_{a,1}\delta^{i,k}\overline1
=L_{a,1}^{i,k}\overline{1}.
\end{align*}
We can rewrite this as a matrix recurrence:
\begin{align}\label{eq:q_recurrence_expanded}
\mathbf q_a(u)\overline1
&=\Biggl(
\mathbf L_{a,1}(u)
-\sum_{b=1}^{a-2}\mathbf L_{a,b+1}(u)\mathbf q_b(u)\notag\\
&\hspace{22mm}
-\bigl(\mathbf L_{a,a}(u)-p(n-a)\mathbf1_p\bigr)
\mathbf q_{a-1}(u)
-p(n-1)\delta_{a,1}\mathbf1_p
\Biggr)\overline1.
\end{align}
Here the sum and the term containing $\mathbf q_{a-1}$ are omitted when $a=1$.
\item \textbf{Identification with the row determinant}. The recurrence
\eqref{eq:q_recurrence_expanded} is equivalent to the following matrix
relation: $$\begin{pmatrix}\mathbf{1}_p&0&...&0\\\mathbf{L}_{2,2}(u)-(n-2)p&\mathbf{1}_p&\ddots&\vdots\\
\vdots&\ddots&\ddots&0\\ 
\mathbf{L}_{n-1,2}(u),&...&\mathbf{L}_{n-1,n-1}(u)-p&\mathbf{1}_p
\end{pmatrix}\begin{pmatrix}
    \mathbf{q}_1(u)\overline{1}\\
    \vdots
\\ \vdots\\ \mathbf{q}_{n-1}(u)\overline{1}
\end{pmatrix}=\begin{pmatrix}
    \mathbf{L}_{1,1}(u)-(n-1)p\\
   \mathbf{L}_{2,1}(u)
\\ \vdots\\\mathbf{L}_{n-1,1}(u)
\end{pmatrix},$$
which uniquely determines the class $$\begin{pmatrix}
    \mathbf{q}_1(u)\overline{1}\\
    \vdots
\\ \vdots\\ \mathbf{q}_{n-1}(u)\overline{1}
\end{pmatrix}=\begin{pmatrix}\mathbf{1}_p&0&...&0\\\mathbf{L}_{2,2}(u)-(n-2)p&\mathbf{1}_p&\ddots&\vdots\\
\vdots&\ddots&\ddots&0\\ 
\mathbf{L}_{n-1,2}(u),&...&\mathbf{L}_{n-1,n-1}(u)-p&\mathbf{1}_p
\end{pmatrix}^{-1}\begin{pmatrix}
    \mathbf{L}_{1,1}(u)-(n-1)p\\
   \mathbf{L}_{2,1}(u)
\\ \vdots\\\mathbf{L}_{n-1,1}(u)
\end{pmatrix}.$$
Substituting this expression into \eqref{eq:quasidet_auxiliary_column}
shows that the reduced quasideterminant is equal to the quasideterminant of
$L^{\mathrm{red}}(u)$:
$$T(u)\cdot \overline{1}=|L^{\mathrm{red}}(u)|_{n,1}=(-1)^{n-1}\mathrm{rdet}(L^{\mathrm{red}}(u)),$$
where the last equality follows from \eqref{quarow}.
\end{itemize}
\end{proof}
\begin{corollary}
\label{cor:bk_rtt} 
\begin{enumerate}
\item \label{prp:quotient_map_hom}
The quotient map restricts to an algebra homomorphism on the invariant
subalgebra:
$$\mathcal{A}^{\mathfrak{m}_n}\rightarrow (\mathcal{A}/\mathcal{I}_\psi)^{\mathfrak{m}_n}.$$
\item
The Brundan--Kleshchev matrix $\mathbf T(u)$ satisfies the RTT relation
with the Yang matrix $R(u)$ of size $p^2\times p^2$:
    $$R(u-v)\mathbf{T}_1(u)\mathbf{T}_2(v)=\mathbf{T}_2(v)\mathbf{T}_1(u)R(u-v).$$
    \end{enumerate}
\end{corollary}
\begin{proof} \begin{itemize}
    \item 
By definition, multiplication in $(\mathcal{A}/\mathcal{I}_\psi)^{\mathfrak{m}_n}$ is given by
$$x\overline{1} \cdot y\overline{1}:=xy\overline{1},$$
where $x,y\in\mathcal{A}$ represent the classes $x\overline{1}$ and $y\overline{1}$ and satisfy
$$[x,m]\overline{1}=[y,m]\overline{1}=0.$$
These conditions ensure that the product is independent of the choice of representatives.
Now let $a,b\in\mathcal{A}^{\mathfrak{m}_n}$. Their images $a\overline{1}$ and $b\overline{1}$ lie in $(\mathcal{A}/\mathcal{I}_\psi)^{\mathfrak{m}_n}$. By the definition of multiplication in this algebra, the quotient map preserves products.
\item
By the first point, the quotient map restricts to an
algebra homomorphism on the invariant subalgebra. Hence the reduced
quasideterminant still satisfies the RTT relation. By
Theorem above~\ref{th:quasidet_and_BK}, it agrees with the
Brundan--Kleshchev generating matrix up to the scalar factor
$(-1)^{n-1}$, which does not affect the RTT relation. Therefore,
$\mathbf T(u)$ satisfies the stated relation.

\end{itemize}
\end{proof}
\section{Differential realization and reduction}
\label{sec:differential_reduction}
\subsection{Differential operators on the lower unipotent group}
Let $N_-\subset\mathrm{GL}_N$ be the lower unipotent
subgroup, and let $\mathfrak n_-$ be its Lie algebra. 
In this section we use a single pair of indices for coordinates and
right-shift operators. Write the coordinate matrix of
$N_-$ as
$$
z=\mathbf1_N+\sum_{1\leq j<i\leq N}z_{i,j}e_{i,j},
\qquad
\mathbb C[N_-]=\mathbb C[z_{i,j}]_{1\leq j<i\leq N}.
$$
\begin{definition}
The algebra $\mathcal D(N_-)$ of polynomial differential operators is
generated by these coordinates and their partial derivatives, with
relations
\begin{align*}
[\partial_{z_{i,j}},z_{k,l}]
&=\delta_{i,k}\delta_{j,l},\\
[\partial_{z_{i,j}},\partial_{z_{k,l}}]
&=[z_{i,j},z_{k,l}]=0.
\end{align*}
\end{definition}

We now recall the standard differential realization of
$\mathfrak{gl}_N$. Let the representation space be a class of functions
on $N_-$, or on an open subset of $N_-$, that is preserved by polynomial
multiplication and differentiation. For $x\in\mathfrak{gl}_N$ and small
$\varepsilon$, the Gauss decomposition has the form
$$
e^{-\varepsilon x}z=z'(\varepsilon)b(\varepsilon),
\qquad z'(\varepsilon)\in N_-,\quad b(\varepsilon)\in B,
$$
where $B\subset\mathrm{GL}_N$ is the upper triangular subgroup.
Fix the diagonal matrices
$$
\lambda=\operatorname{diag}(\lambda_1,\dots,
\lambda_N) \ \ ;\ \ \rho =\operatorname{diag}(N-1,N-2,\ldots,1,0).
$$
For $b\in B$, define the twist
$$
\chi_\lambda(b)=\exp\bigl(-\operatorname{tr}((\lambda+\rho)\log b)\bigr).
$$
For generic $\lambda$, this expression depends on a choice of branch of
the logarithm, but its infinitesimal action does not. We denote the
realized operator corresponding to $x$ by $E_x$. Define the left action by
\begin{align}\label{eq:left_action}
E_x \Phi(z)
=\left.\frac{d}{d\varepsilon}
\left(\chi_\lambda(b(\varepsilon))
\Phi(z'(\varepsilon))\right)\right|_{\varepsilon=0}.
\end{align}
If $x\in\mathfrak n_-$, then $b(\varepsilon)=\mathbf1_N$; hence the
corresponding operators $E_x$ are independent of $\lambda$.
\begin{align}\label{eq:left_action-}
E_x\Phi(z)
=\left.\frac{d}{d\varepsilon}\Phi(z-\varepsilon xz)
\right|_{\varepsilon=0} = -\sum_{i>j} (xz)_{i,j}\,
\partial_{z_{i,j}}\,\Phi(z)
\end{align}
The standard basis of $\mathfrak{gl}_N$ consists of the matrix units
$e_{i,j}$. It follows from definition \eqref{eq:left_action} that the map
is linear:
\begin{align}\label{basis}
x = \sum_{i,j} x_{i,j} e_{i,j} \rightarrow E_{x} = \sum_{i,j} x_{i,j} E_{e_{i,j}}
\end{align}
For $x=e_{i,j}$ we denote the corresponding operator by $E_{i,j} = E_{e_{i,j}}$. 
The generators naturally split into two classes. Those with $i>j$
correspond to elements $e_{i,j}\in\mathfrak n_-$ and do not depend on
$\lambda$:
\begin{align}\label{E}
E_{i,j}=
-\sum_{p>s}\left(e_{i,j}\mathrm{z}\right)_{p,s}\partial_{z_{p,s}} =
-\partial_{z_{i,j}} - \sum_{m=1}^{j-1}
z_{j,m}\,\partial_{z_{i,m}}\,,
\end{align}
The generators $E_{i,j}$ with $i\leq j$, by contrast, depend on
$\lambda$. Direct expressions for these generators analogous to
\eqref{E} would be cumbersome. However, a simple formula expresses all of
them in terms of the right-shift generators.

For $x\in\mathfrak n_-$, define the right-shift operator by
\begin{align}\label{eq:right_action}
D_x\Phi(z)
=\left.\frac{d}{d\varepsilon}\Phi(ze^{\varepsilon x})
\right|_{\varepsilon=0} = \sum_{i>j} (zx)_{i,j}\,
\partial_{z_{i,j}}\,\Phi(z).
\end{align}
Both \eqref{eq:left_action} and \eqref{eq:right_action} define polynomial
differential operators. For $x=e_{i,j}$ with $i>j$, we denote the
corresponding right-shift operator by $D_{i,j}$. Any right-shift operator
is a linear combination of these basis operators, just as in
\eqref{basis}. The individual right-shift generators have the following
explicit form:
\begin{align}\label{eq:explicit_right_shifts}
D_{i,j}
=\sum_{p>s} (ze_{i,j})_{p,s}\,
\partial_{z_{p,s}} = \partial_{z_{i,j}}
+\sum_{k=i+1}^N z_{k,i}\,\partial_{z_{k,j}},
\qquad i>j.
\end{align}

\begin{corollary}
The generators $E_{i,j}$ obey the commutation relations
\eqref{eq:comm_relations_gl}, and the left action \eqref{eq:left_action}
extends to an algebra homomorphism
$$\rho_\lambda:
\mathcal U(\mathfrak{gl}_N)\longrightarrow\mathcal D(N_-).
$$
\end{corollary}

\paragraph{Example.}
Consider the familiar case $N=2$. Then $N_-$ has one
coordinate $x=z_{2,1}$. Take the matrix unit $e_{1,2}$ and
$\lambda+\rho =\operatorname{diag}(\lambda_1+1,\lambda_2)$. The Gauss
decomposition is
\begin{equation*}
\left(\begin{array}{cc}
  1 & -\varepsilon \\
  0 & 1 \\
\end{array}\right)\left(\begin{array}{cc}
  1 & 0 \\
  x & 1 \\
\end{array}\right) = \left(\begin{array}{cc}
  1 & 0 \\
  \frac{x}{1-\varepsilon x} & 1 \\
\end{array}\right)\left(\begin{array}{cc}
  {1-\varepsilon x} & -\varepsilon \\
  0 & \frac{1}{1-\varepsilon x} \\
\end{array}\right).
\end{equation*}
Thus
$$
\chi_\lambda(b(\varepsilon))
=\exp\bigl((\lambda_2-\lambda_1-1)\log(1-\varepsilon x)\bigr),
$$
and \eqref{eq:left_action} gives
$$
E_{1,2}\Phi(x)
=\left.\frac{d}{d\varepsilon}
\left[
\Phi\left(\frac{x}{1-\varepsilon x}\right)
e^{(\lambda_2-\lambda_1-1)\log(1-\varepsilon x)}
\right]\right|_{\varepsilon=0}
=\bigl(x^2\partial_x+(\lambda_1-\lambda_2+1)x\bigr)\Phi(x).
$$
Similarly,
$$
E_{2,1}=-\partial_x,\qquad
E_{1,1}=x\partial_x+\lambda_1+1,\qquad
E_{2,2}=-x\partial_x+\lambda_2.
$$
The right-shift generator in this case is $D_{2,1}=\partial_x$.

\begin{corollary} \label{cor:left_and_right_invariance}
For all $x,y\in\mathfrak n_-$, the left- and right-shift operators
commute:
\begin{align} \label{eq:left_and_right_invariance}
[E_x,D_y]=0.
\end{align}
\end{corollary}
\begin{proof}
It is easy to check that 
\begin{align*}
E_x\,D_y\Phi(z) = D_y\,E_x\Phi(z) = 
\left.\frac{d^2}{d\varepsilon_1 d\varepsilon_2}\Phi(e^{-\varepsilon_1 x}ze^{\varepsilon_2 y})
\right|_{\varepsilon_1=\varepsilon_2=0}.
\end{align*}
\end{proof}
We now give a compact formula for the realized universal matrix $E$.
Assemble the
right-shift operators into the strictly upper-triangular matrix
\begin{align}\label{eq:differential_right_shifts}
D&=\sum_{1\leq j<i\leq N}e_{j,i}\otimes D_{i,j}.
\end{align}

\begin{proposition}[Derkachov--Manashov \cite{DM,DM1}]
The matrix of left-action generators satisfies the factorized formula
\begin{equation}\label{eq:differential_realization}
E=z(\lambda-D)z^{-1}.
\end{equation}
Equivalently, the image of the $L$-operator is
\begin{equation}\label{eq:differential_L_operator}
L(u)
=z(u+\lambda-D)z^{-1}.
\end{equation}
\end{proposition}

\begin{remark}
We would like to emphasize that this formula can be interpreted as a quantization of the Springer map
\cite{BeilinsonBernstein}.
\end{remark}
\begin{remark}
It is instructive to present everything in the visual matrix form
\begin{align*}
E = \begin{pmatrix}1&0&0&\ldots&0\\
z_{2,1}&1&0&\ldots&0\\
z_{3,1}&z_{3,2}&1&\ldots&0\\
\vdots&\vdots&\vdots&\ddots&\vdots\\
z_{N,1}&z_{N,2}&z_{N,3}&\ldots &1
\end{pmatrix}
\begin{pmatrix} \lambda_1 &-D_{2,1}&
-D_{3,1}&\ldots&-D_{N,1}\\
0&\lambda_2& -D_{3,2}&\ldots& -D_{N,2}\\
0&0&\lambda_3&\ldots& -D_{N,3}\\
\vdots&\vdots&\vdots& \ddots& \vdots\\
0&0&0&\ldots&\lambda_{N}
 \end{pmatrix}
\left(\begin{matrix}1&0&0&\ldots&0\\
z_{2,1}&1&0&\ldots&0\\
z_{3,1}&z_{3,2}&1&\ldots&0\\
\vdots&\vdots&\vdots&\ddots&\vdots\\
z_{N,1}&z_{N,2}&z_{N,3}&\ldots &1
\end{matrix}\right)^{-1}
\end{align*}
\end{remark}

\begin{proof}
Let $\mathfrak{gl}_N=\mathfrak n_-\oplus\mathfrak b$,
where $\mathfrak b$ is the upper triangular subalgebra. For
$x\in\mathfrak{gl}_N$, write
$$
x=x_-+x_+,
\qquad x_-\in\mathfrak n_-,\quad x_+\in\mathfrak b,
$$
We first consider an element
$e^{\varepsilon a}=\mathbf1_N+\varepsilon a+O(\varepsilon^2)$
in a small neighborhood of the identity. Put
$x=z^{-1}az$. To first order in $\varepsilon$, the Gauss decomposition
used in \eqref{eq:left_action} takes the form
\begin{align*}
z'(\varepsilon)
&=z(\mathbf1_N-\varepsilon x_-)+O(\varepsilon^2),\\
b(\varepsilon)
&=\mathbf1_N-\varepsilon x_++O(\varepsilon^2),
\end{align*}
which could be checked directly: 
\begin{align*}
e^{-\varepsilon a}z = (\mathbf1_N-\varepsilon a)z +O(\varepsilon^2)
=z(1-\varepsilon z^{-1}az)+O(\varepsilon^2) \\ 
z'(\varepsilon)b(\varepsilon)=z(\mathbf1_N-\varepsilon x_-)(\mathbf1_N-\varepsilon x_+)+O(\varepsilon^2) = z(1-\varepsilon (x_{-}+x_{+}))+O(\varepsilon^2)
\end{align*}
Therefore the infinitesimal action can be written as
\begin{equation}\label{eq:differential_infinitesimal_action}
E_a =\sum_{k=1}^N (\lambda_k+\rho_k)(x_+)_{k,k}-D_{x_-},
\qquad x=z^{-1}az.
\end{equation}
Here $\lambda+\rho =\operatorname{diag}
(\lambda_1+\rho_1,\ldots,\lambda_N+\rho_N)$.
For $a=e_{i,j}$ we have explicit expressions  
\begin{align}
(x_{-})_{k,l} = z_{j,l} \left(z^{-1}\right)_{k,i}\ \ ; \ \ k > l 
\ \ ;\ \ 
(x_{+})_{k,k} = z_{j,k} \left(z^{-1}\right)_{k,i}
\end{align}
so that 
\begin{multline*}
E_{i,j} =\sum_{k=1}^N z_{j,k}(\lambda_k+\rho_k)\left(z^{-1}\right)_{k,i} - 
\sum_{k>l} z_{j,l} \left(z^{-1}\right)_{k,i} D_{k,l} = \\
\left(z(\lambda+\rho) z^{-1}\right)_{j,i} - \sum_{k>l} z_{j,l} D_{k,l} \left(z^{-1}\right)_{k,i}
+ \sum_{k>l} z_{j,l} \left[D_{k,l}\,,\left(z^{-1}\right)_{k,i}\right]
\end{multline*}
It remains to calculate the commutator. From the
definition of the right action, we have
\begin{align*}
D_x z = z D_x + z x \rightarrow 
z^{-1} D_x=(D_x+x)z^{-1} \rightarrow \left[D_x\,,z^{-1}\right] = - x z^{-1}
\end{align*}
or in explicit index form 
\begin{align*}
[D_{k,l}\,,\left(z^{-1}\right)_{j,i}] = 
-\left(e_{k,l} z^{-1}\right)_{j,i} = -\delta_{k,j}\left(z^{-1}\right)_{l,i}
\end{align*}
and consequently 
\begin{align*}
\sum_{k>l}  z_{j,l} \left[D_{k,l}\,,\left(z^{-1}\right)_{k,i}\right] = 
- \sum_{l} (N-l) z_{j,l} \left(z^{-1}\right)_{l,i} = -\left(z\rho z^{-1}\right)_{j,i}
\end{align*}
Collecting everything together we obtain 
\begin{align}\label{Eij}
E_{i,j} = 
\left(z \lambda z^{-1}\right)_{j,i} - \sum_{k>l} z_{j,l} D_{k,l} \left(z^{-1}\right)_{k,i}
\end{align}
This formula is exactly the formula \eqref{eq:differential_realization} 
in index notations.
\end{proof}

\subsection{Block notation for differential matrices}
We use $n\times n$ block matrices with blocks of size $p\times p$.
Decompose the diagonal parameter matrix into blocks as
$$
\lambda=\operatorname{diag}(\boldsymbol\lambda_1,\ldots,
\boldsymbol\lambda_n),\qquad
\boldsymbol\lambda_a=\operatorname{diag}
(\lambda_{(a-1)p+1},\ldots,\lambda_{ap}).
$$
In the same block notation, the upper-triangular right-shift matrix is
$$
 D=
\begin{pmatrix}
\mathbf D_1&\mathbf D_{1,2}&\mathbf D_{1,3}&\cdots&\mathbf D_{1,n}\\
0&\mathbf D_2&\mathbf D_{2,3}&\cdots&\mathbf D_{2,n}\\
0&0&\mathbf D_3&\cdots&\mathbf D_{3,n}\\
\vdots&\vdots&\vdots&\ddots&\vdots\\
0&0&0&\cdots&\mathbf D_n
\end{pmatrix},
$$
where the off-diagonal blocks and diagonal blocks are given, respectively,
by
$$
\begin{aligned}
\mathbf D_{a,b}
&=\sum_{i,j=1}^p e_{j,i}\otimes
D_{(b-1)p+i,\,(a-1)p+j},&& a<b,\\
\mathbf D_a
&=\sum_{1\leq j<i\leq p}e_{j,i}\otimes
D_{(a-1)p+i,\,(a-1)p+j},&&1\leq a\leq n.
\end{aligned}
$$
The coordinate matrix $z$ has the block form
$$
z=
\begin{pmatrix}
\mathbf z_1&0&0&\cdots&0\\
\mathbf z_{2,1}&\mathbf z_2&0&\cdots&0\\
\mathbf z_{3,1}&\mathbf z_{3,2}&\mathbf z_3&\cdots&0\\
\vdots&\vdots&\vdots&\ddots&\vdots\\
\mathbf z_{n,1}&\mathbf z_{n,2}&\mathbf z_{n,3}&\cdots&\mathbf z_n
\end{pmatrix}.
$$
For $a>b$, the entries of the off-diagonal coordinate blocks are
$$
(\mathbf z_{a,b})_{i,j}
=z_{(a-1)p+i,\,(b-1)p+j},\qquad 1\leq i,j\leq p.
$$
Each diagonal block $\mathbf z_a$ is lower unitriangular, with entries
$$
(\mathbf z_a)_{i,j}=
z_{(a-1)p+i,\,(a-1)p+j}
$$
We also record useful block-factorization formulas in terms of the usual
derivatives $\partial_{z_{i,j}}$.
It is useful to introduce block matrices of derivatives
\begin{align}
\mathbf{d}_{a,b}^{i,j} = \partial_{z_{j+p(b-1),i+p(a-1)}}
\ \ ;\ \ 
\mathbf{d}_{a}^{i,j} = \left(\mathbf{d}_{a,a}\right)_{+}^{i,j} = \begin{cases} \partial_{z_{j+p(a-1),i+p(a-1)}}
\ \ ;\ \ i<j \\ 
0\ \ ;\ \ i=j \\ 
0\ \ ;\ \ i>j
\end{cases}
\end{align}

\begin{proposition}
\label{block}
For $a<b$, the off-diagonal blocks $\mathbf{D}_{a,b}$ of $D$ and
$\mathbf{E}_{a,b}$ of the generator matrix (see
\eqref{eq:generator_blocks}) are expressed in terms of $\mathbf{d}_{a,b}$
and $\mathbf{z}_a$ as follows. The diagonal blocks of $D$ are given by:
\begin{align}\label{DLd}
\mathbf{D}_{a,b} = \mathbf{d}_{a,b}\mathbf{z}_{b}+
\sum_{c=b+1}^n \mathbf{d}_{a,c}\mathbf{z}_{c,b} \ \ ;\ \ 
\mathbf{E}_{a,b} = - \mathbf{z}_{a}\mathbf{d}_{a,b}-
\sum_{c=1}^{a-1} \mathbf{z}_{a,c}\mathbf{d}_{c,b} \\ 
\label{Dad}
\mathbf{D}_{a} = \left(\mathbf{d}_{a}\mathbf{z}_{a}\right)_{+}+
\sum_{c=a+1}^n \left(\mathbf{d}_{a,c}\mathbf{z}_{c,a}\right)_{+}
\end{align}
\end{proposition}
The proof is a direct check.
\begin{remark}
Note the following matrix relation
\begin{align*}
\left(\begin{pmatrix} \mathbf{D}_1 & \mathbf{D}_{1,2}&
\mathbf{D}_{1,3}&\ldots& \mathbf{D}_{1,n}\\
0&\mathbf{D}_2& \mathbf{D}_{2,3}&\ldots& \mathbf{D}_{2,n}\\
0&0&\mathbf{D}_3&\ldots& \mathbf{D}_{3,n}\\
\vdots&\vdots&\vdots& \ddots& \vdots\\
0&0&0&\ldots&\mathbf{D}_{n}
 \end{pmatrix}
\begin{pmatrix}\mathbf{z}_1&0&0&\ldots&0\\
\mathbf{z}_{2,1}&\mathbf{z}_2&0&\ldots&0\\
\mathbf{z}_{3,1}&\mathbf{z}_{3,2}&\mathbf{z}_3&\ldots&0\\
\vdots&\vdots&\vdots&\ddots&\vdots\\
\mathbf{z}_{n,1}&\mathbf{z}_{n,2}&\mathbf{z}_{n,3}&\ldots &\mathbf{z}_n
\end{pmatrix}^{-1}\right)_{+} = 
\begin{pmatrix} 0 & \mathbf{d}_{1,2}&
\mathbf{d}_{1,3}&\ldots& \mathbf{d}_{1,n}\\
0&0& \mathbf{d}_{2,3}&\ldots& \mathbf{d}_{2,n}\\
0&0&0&\ldots& \mathbf{d}_{3,n}\\
\vdots&\vdots&\vdots& \ddots& \vdots\\
0&0&0&\ldots&0
 \end{pmatrix}
\end{align*}
follows from the formulas in \eqref{DLd}--\eqref{Dad}. It expresses the
strictly upper block-triangular part of the generator matrix $E$ in the
form
\begin{align}\label{Eplus}
E_{+} = -\left(\begin{pmatrix} \mathbf{z}_1&0&0&\ldots&0\\
\mathbf{z}_{2,1}&\mathbf{z}_2&0&\ldots&0\\
\mathbf{z}_{3,1}&\mathbf{z}_{3,2}&\mathbf{z}_3&\ldots&0\\
\vdots&\vdots&\vdots&\ddots&\vdots\\
\mathbf{z}_{n,1}&\mathbf{z}_{n,2}&\mathbf{z}_{n,3}&\ldots &\mathbf{z}_n
\end{pmatrix}
\begin{pmatrix} 0 & \mathbf{d}_{1,2}&
\mathbf{d}_{1,3}&\ldots& \mathbf{d}_{1,n}\\
0&0& \mathbf{d}_{2,3}&\ldots& \mathbf{d}_{2,n}\\
0&0&0&\ldots& \mathbf{d}_{3,n}\\
\vdots&\vdots&\vdots& \ddots& \vdots\\
0&0&0&\ldots&0
\end{pmatrix}\right)_{+}
\end{align}
which is equivalent to the formula for $\mathbf{E}_{a,b}$ for $a<b$ in \eqref{DLd}.
\end{remark}
\subsection{Quasideterminant with differential entries}

\begin{proposition}\label{prop:differential_quasideterminant}
The
corner quasideterminant of $L(u)$ can be written entirely in terms of the
diagonal coordinate blocks, the diagonal parameter blocks, and the
right-shift blocks:
\begin{equation}\label{eq:differential_quasideterminant_block_form}
\left|L(u)\right|_{n,1}
=\left|
\begin{pmatrix}
(u+\boldsymbol\lambda_1-\mathbf D_1)\mathbf z_1^{-1}
&-\mathbf D_{1,2}&-\mathbf D_{1,3}&\cdots&-\mathbf D_{1,n}\\
0&u+\boldsymbol\lambda_2-\mathbf D_2&-\mathbf D_{2,3}&\cdots&-\mathbf D_{2,n}\\
0&0&u+\boldsymbol\lambda_3-\mathbf D_3&\cdots&-\mathbf D_{3,n}\\
\vdots&\vdots&\vdots&\ddots&\vdots\\
0&0&0&\cdots&\mathbf z_n(u+\boldsymbol\lambda_n-\mathbf D_n)
\end{pmatrix}
\right|_{n,1}.
\end{equation}
\end{proposition}
\begin{proof}
The block entries of $z$ give the following explicit factorization:
\begin{align*}
z=
\begin{pmatrix}
\mathbf 1_p&0&0&\cdots&0\\
\mathbf z_{2,1}\mathbf z_1^{-1}&\mathbf 1_p&0&\cdots&0\\
\mathbf z_{3,1}\mathbf z_1^{-1}&\mathbf z_{3,2}\mathbf z_2^{-1}&\mathbf 1_p&\cdots&0\\
\vdots&\vdots&\vdots&\ddots&\vdots\\
\mathbf z_{n,1}\mathbf z_1^{-1}&\mathbf z_{n,2}\mathbf z_2^{-1}&
\mathbf z_{n,3}\mathbf z_3^{-1}&\cdots&\mathbf 1_p
\end{pmatrix}\cdot
\begin{pmatrix}
\mathbf z_1&0&0&\cdots&0\\
0&\mathbf z_2&0&\cdots&0\\
0&0&\mathbf z_3&\cdots&0\\
\vdots&\vdots&\vdots&\ddots&\vdots\\
0&0&0&\cdots&\mathbf z_n
\end{pmatrix}.
\end{align*}
Substituting this factorization into $L(u)=z(u+\lambda-D)z^{-1}$,
the first lower block-unitriangular factor and its inverse flank the
conjugate of $u+\lambda-D$ by the displayed diagonal block matrix.
The first factor and its inverse have the forms allowed in
\eqref{XAY}. Applying that invariance therefore removes these two
factors, leaving the corner quasideterminant of the block-conjugated
upper-triangular matrix displayed below:

$$\left|
\begin{pmatrix}
\mathbf z_1&0&0&\cdots&0\\
0&\mathbf z_2&0&\cdots&0\\
0&0&\mathbf z_3&\cdots&0\\
\vdots&\vdots&\vdots&\ddots&\vdots\\
0&0&0&\cdots&\mathbf z_n
\end{pmatrix}\left(u+\lambda-
\begin{pmatrix}
\mathbf D_1&\mathbf D_{1,2}&\mathbf D_{1,3}&\cdots&\mathbf D_{1,n}\\
0&\mathbf D_2&\mathbf D_{2,3}&\cdots&\mathbf D_{2,n}\\
0&0&\mathbf D_3&\cdots&\mathbf D_{3,n}\\
\vdots&\vdots&\vdots&\ddots&\vdots\\
0&0&0&\cdots&\mathbf D_n
\end{pmatrix}\right)
\begin{pmatrix}
\mathbf z_1&0&0&\cdots&0\\
0&\mathbf z_2&0&\cdots&0\\
0&0&\mathbf z_3&\cdots&0\\
\vdots&\vdots&\vdots&\ddots&\vdots\\
0&0&0&\cdots&\mathbf z_n
\end{pmatrix}^{-1}\right|_{n,1}$$

For the second application of \eqref{XAY}, use the following explicit
left and right multipliers, respectively:
$$
\begin{pmatrix}
\mathbf z_1&0&0&\cdots&0\\
0&\mathbf z_2&0&\cdots&0\\
0&0&\mathbf z_3&\cdots&0\\
\vdots&\vdots&\vdots&\ddots&\vdots\\
0&0&0&\cdots&\mathbf 1_p
\end{pmatrix},
\qquad
\begin{pmatrix}
\mathbf 1_p&0&0&\cdots&0\\
0&\mathbf z_2^{-1}&0&\cdots&0\\
0&0&\mathbf z_3^{-1}&\cdots&0\\
\vdots&\vdots&\vdots&\ddots&\vdots\\
0&0&0&\cdots&\mathbf z_n^{-1}
\end{pmatrix}.
$$
Their required principal blocks are invertible, and they have the forms
prescribed in \eqref{XAY}. Factoring out the left matrix from the left, and right matrix from the right, we arrive to the formula \eqref{eq:differential_quasideterminant_block_form}.
\end{proof}
\subsection{Reduction of the algebra of differential operators}
Let us proceed with algebraic reduction of the differential realization. 
Let $M_n=\exp(\mathfrak m_n)$ be the block-lower unipotent subgroup,
let $N_{\mathrm d}$ be the block-diagonal subgroup of $N_-$, and let
$N_-'\subset\mathrm{GL}_p$ be the lower unipotent subgroup.
Every $z\in N_-$ has a unique factorization
\begin{equation}\label{eq:differential_coordinate_factorization}
z=mz_{\mathrm d},\qquad
m\in M_n,\qquad
z_{\mathrm d}=\operatorname{diag}(\mathbf z_1,\dots,\mathbf z_n),
\qquad \mathbf z_a\in N_-'.
\end{equation}
Indeed, the diagonal blocks of $z$ determine $z_{\mathrm d}$, after which
\begin{align}
m=zz_{\mathrm d}^{-1} \,.
\end{align} 
Thus multiplication gives an algebraic
isomorphism $M_n\times N_{\mathrm d}\simeq N_-$, with
$N_{\mathrm d}\simeq(N_-')^n$.

Use the character $\psi$ from \eqref{eq:quantum_character} and set
$$
I:=\mathcal D(N_-)\bigl\langle x-\psi(x)
\mid x\in\mathfrak m_n\bigr\rangle,
\qquad
\mathcal DQ=\mathcal D(N_-)/I.
$$
We define the reduced algebra by
\begin{equation}\label{eq:differential_reduced_algebra}
\mathcal D_-
=(\mathcal DQ)^{\mathfrak m_n},
\end{equation}
i.e. the elements of the quotient $\mathcal DQ$ that are invariant under
the action of $\mathfrak m_n$.

\begin{proposition}\label{prp:differential_reduction}
There is a natural algebra isomorphism
$$
\xi:\mathcal D_-\longrightarrow \mathcal D(N_-')^{\otimes n}.
$$
The reduced algebra is generated by the entries of the diagonal
coordinate blocks $\mathbf z_a$ and by the classes of the right-shift
operators $D_{(a-1)p+i,\,(a-1)p+j}$, $i>j$.
\end{proposition}
\begin{proof}
\begin{itemize}\item
First, we verify that the diagonal-block right-shift operators
$D_{(a-1)p+i,(a-1)p+j}$ and the diagonal coordinates
$z_{(a-1)p+i,(a-1)p+j}$ are invariant under $\mathfrak m_n$.
By Corollary~\ref{cor:left_and_right_invariance}, the right-shift
operators commute with the infinitesimal left action of $\mathfrak m_n$.
The diagonal factor $z_{\mathrm d}$ is invariant as well: left
multiplication by $M_n$ changes the factor $m$ but leaves the diagonal
blocks of $z$ unchanged. In the coordinates $(m,z_{\mathrm d})$, the
action of $a\in M_n$ is therefore
$$z\mapsto az\Rightarrow m\mapsto am,\quad z_\mathrm{d}\mapsto z_\mathrm{d}. $$

\item
To describe the quotient, we order the generators so that the diagonal-block
right shifts precede the infinitesimal action generators from $\mathfrak m_n$.
In this order, every element has a unique decomposition as a sum of terms
of the form
$$P(z)D_{(a_1-1)p+i_1,(a_1-1)p+j_1}\cdot\ldots\cdot
D_{(a_k-1)p+i_k,(a_k-1)p+j_k}\cdot
E_{a_{k+1},b_{k+1}}^{i_{k+1},j_{k+1}}\cdot\ldots\cdot
E_{a_{k+s},b_{k+s}}^{i_{k+s},j_{k+s}},$$
where $P(z)\in\mathbb C[N_-]$ is a polynomial function and the
$E_{a_r,b_r}^{i_r,j_r}$ are the images of generators of $\mathfrak m_n$.
Replacing each such generator by its shifted form
$E_{a_r,b_r}^{i_r,j_r}-\psi(E_{a_r,b_r}^{i_r,j_r})$ separates terms with no
shifted action generator from terms containing at least one. The latter lie
in $I$, while uniqueness of the ordered decomposition makes the resulting
sum direct:
 $$
\mathcal D(N_-)
=\mathbb C[N_-]\cdot\mathcal D
\oplus I,
$$
where $\mathcal D$ is generated by the diagonal-block
right-shift operators $D_{(a-1)p+i,(a-1)p+j}$ 

and the diagonal-block
coordinates $z_{(a-1)p+i,(a-1)p+j}$.
Consequently,
$$
\mathcal DQ
\cong
\mathbb C[N_-]\cdot \mathcal{D}.
$$
This is an isomorphism of vector spaces. In the factorization
$z=mz_{\mathrm d}$, the left action of $M_n$ acts on $m$ and fixes
$z_{\mathrm d}$; hence invariant coefficient functions depend only on
$z_{\mathrm d}$. Therefore, using
$$
\mathbb C[N_-]^{\mathfrak m_n}
\cong\mathbb C[N_{\mathrm d}]
$$
we identify the reduction $\mathcal{D}_-$ with the subalgebra $\mathcal{D}$
of $\mathcal{D}(N_-)$ generated by the block-diagonal right shifts
$D_{(a-1)p+i,(b-1)p+j}$ and the block-diagonal coordinates
$z_{(a-1)p+i,(b-1)p+j}$.
\item Finally, we identify $\mathcal{D}$ with
$\mathcal{D}(N_-')^{\otimes n}$. We note that the previous point can be understood as follows: for any element of the reduced algebra $\mathcal{D}_-$, we can find a unique representative in the algebra $\mathcal{D}$. In particular, the elements $E_{a+1,a}^{i,j}$ have representative $\delta_{i,j}$ in $\mathcal{D}$.

The latter algebra $\mathcal{D}$ can be identified with differential algebra in diagonal variables
$\mathcal{D}(N_{\mathrm{d}})$ by considering the action of the diagonal-block
right shifts on functions of the diagonal-block coordinates. Their
expressions on $N_-$ may also contain derivatives in off-diagonal
coordinates, but these derivatives vanish on such functions.

Thus, their
restrictions give the corresponding right shifts on $N_{\mathrm{d}}$.
Conversely, any differential operator on $N_{\mathrm{d}}$ extends to a
differential operator on $N_-$ by expressing it in the basis of right
shifts and extending those shifts via formula \eqref{Dad}. The polynomial
prefactor of an operator of the form $P(z_{\mathrm{d}})D_{i_1,j_1}\cdots D_{i_n,j_n}$ is unchanged under this
extension. Thus, we define mutually inverse homomorphisms.

Thus, we have constructed a map from $\mathcal{D}_-$ to $\mathcal{D}(N_{\mathrm{d}})$ in two steps --- firstly by mapping to $\mathcal{D}$ via finding a good representative, and then setting to zero all the off-diagonal variables in the operators $D_{i,j}$.

Since $N_\mathrm{d}\simeq(N_-')^n$, this identification extends
naturally to the differential algebras:
\begin{align*}
\xi: \ &\mathcal{D}_-\longrightarrow \mathcal{D}(N_-')^{\otimes n}
\\  &\mathbf z_a\longmapsto z_{[a]},\qquad
D_{(a-1)p+i,\,(a-1)p+j}\longmapsto D_{i,j,[a]}=\partial_{z_{i,j,[a]}}+\sum_{k>i}^pz_{k,i,[a]}\partial_{z_{k,j,[a]}},
\end{align*}
where the square-bracketed subscripts label the tensor factors of the $\mathcal{D}(N_-')^{\otimes n}$. 
\end{itemize}
\end{proof}
\begin{remark}
The map $\xi$ identifies the diagonal shift block $\mathbf{D}_a$ with
the matrix $D_{[a]}$ whose $(j,i)$-entry is $D_{i,j,[a]}$ for $i>j$:
\begin{equation}\label{eq:differential_tensor_notation}
\begin{aligned}\xi:
\mathbf{D}_{a}
\longmapsto D_{[a]}.
\end{aligned}
\end{equation}
In terms of the differential generators, this identification reads
\begin{align}\label{Da}
\xi(\mathbf{D}_{a})=\left({d}_{[a]}{z}_{[a]}\right)_{+},
\end{align}
where $(d_{[a]})_{i,j}=\partial_{z_{j,i,[a]}}$ is the matrix of partial
derivatives in the $a$-th tensor factor. Thus, upon restriction to functions
of the diagonal-block coordinates, the terms involving off-diagonal blocks
in formula \eqref{Dad} do not contribute.
\end{remark}

The reduction of $\rho_\lambda$ therefore defines the homomorphism
$$
\rho_\lambda^{\mathrm{red}}:\mathcal W_n(\mathfrak{gl}_p)
\longrightarrow\mathcal D_-.
$$ 
Composing with $\kappa$ identifies its target with
$\mathcal D(N_-')^{\otimes n}$. However, the choice of the nontrivial
character $\psi$ imposes nontrivial constraints on the off-diagonal
right-shift generators. We express the reduction constraints
$x-\psi(x)$ in terms of the generators $D_{i,j}$ and $z_{i,j}$.

\begin{proposition} \label{prp:reduction_of_D}
For $a<b$, the generators $D_{(b-1)p+i,\,(a-1)p+j}$ satisfy
\begin{align}
&D_{(b-1)p+i,\,(a-1)p+j}\overline{1}
=-\delta_{b,a+1}(\mathbf{z}_a^{-1}\mathbf{z}_b)_{j,i}
\end{align}
In the block notation, this is equivalent to
\begin{align}\label{eq:reduction_of_shifts}
&\mathbf{D}_{a,b}\overline{1}
=-\delta_{b,a+1}\mathbf{z}_a^{-1}\mathbf{z}_b,
\qquad a<b.
\end{align}
\end{proposition}
\begin{proof}
As before, denote by $(~)_+$ the projection onto the strictly upper block-triangular part.
The formulas \eqref{Dad},\eqref{Eplus} imply the relation
$$
(Dz^{-1})_+=-(z^{-1}E_+)_+,
$$
where $E_+=(E)_+$ is the strictly upper block-triangular part of the generator matrix. We now
apply this relation to $\overline{1}$ and obtain
$$
(Dz^{-1})_+\overline{1}=-\left(z^{-1}\begin{pmatrix}
    0&\mathbf{1}_p&0&\ldots&0\\
    0&0&\mathbf{1}_p&\ddots&\vdots\\
    \vdots&\ddots&\ddots&\ddots&0\\
    0&\ldots&\ldots&0&\mathbf{1}_p\\
    0&\ldots&\ldots&0&0
\end{pmatrix}\right)_+\overline{1},$$
and solve the resulting relation iteratively to obtain
\eqref{eq:reduction_of_shifts}.
\end{proof}
However, all the abstract relations and identifications (such as the map $\xi$) can be understood directly in terms of manipulations with the functional spaces --- similarly to the construction of the map $\rho_\lambda$. We explain this point of view on the reduction of the differential algebra in Appendix~\ref{app:func}.

\subsection{Reduction of the differential quasideterminant}
Assume that we have chosen a localization $\mathcal{DN}$ of the
differential-operator algebra in which the quasideterminant with
differential-operator entries is defined, and that the homomorphism
$\rho_\lambda$ extends to a homomorphism of the localized algebras,
denoted by $\widehat\rho_\lambda$. The image of the
quasideterminant under this extension is the quasideterminant of the
matrix of generators in \eqref{eq:differential_realization}; this matrix is
$L(u)=z(u+\lambda-D)z^{-1}$.

Both algebras have left ideals generated by $m-\psi(m)$, for
$m\in\mathfrak m_n$. The quasideterminant with differential entries
belongs to $\mathcal{DN}^{\mathfrak m_n}$. We denote the corresponding
left ideal in the localized differential algebra by $\mathcal{DI}$, so
that the quotient is $\mathcal{DN}/\mathcal{DI}$ and the reduced
differential algebra is its invariant subalgebra
$(\mathcal{DN}/\mathcal{DI})^{\mathfrak m_n}$.
The homomorphism $\widehat\rho_\lambda$ therefore descends to a
homomorphism of the invariant reductions:
\begin{align*}
\widehat\rho_\lambda:
\mathcal{A}&\longrightarrow
 \mathcal{DN},\\
\widehat\rho_\lambda^{\mathrm{red}}:
(\mathcal{A}/\mathcal{I}_\psi)^{\mathfrak{m}_n}&\longrightarrow (\mathcal{DN}/\mathcal{DI})^{\mathfrak{m}_n},
\end{align*}
where $\mathcal{DI}$ is the left ideal implementing the reduction
constraints in the localized differential algebra.

\begin{remark}
    The appropriate construction of the rings \(\mathcal{A}\) and \(\mathcal{DN}\) is by Rees localization.
    One can introduce the Kazhdan filtration on \(\mathcal{U}(\mathfrak{gl}_N)\):
    \begin{equation}\deg E_{a,b}^{i,j}=1+b-a \end{equation}
     and the corresponding  filtration on \(\mathcal{D}(N_-)\):
     
\begin{equation}\label{eq:diff_Kazhdan_weights}
\deg z_{(a-1)p+i,(b-1)p+j}=a-b,\qquad
\deg\partial_{z_{(a-1)p+i,(b-1)p+j}}=1+b-a.
\end{equation}
    The homomorphism \(\rho_\lambda\) extends to a homomorphism between the filtered algebras.
    The ideals \({I}_\psi\) and \(I\) are filtered, so the quotient modules inherit induced filtrations.
    Thus, one can consider the Rees rings and the Rees modules of these quotients (see \cite{DSKV}).
    The homomorphism \(\rho_\lambda\) therefore extends to the Rees rings and Rees modules.
    Rees localization is constructed by successively adding formal series that become invertible after taking the quotient.
    In this way, the map \(\rho_\lambda\) extends to the localized rings $\hat{\rho}_\lambda$.
    It follows from the results of \cite{DSKV} that all the used quasideterminants are well-defined in the Rees localization -- since all the matrices in the quasideterminants become invertible after reduction.
    Hence, all the formulas are justified by considering this concrete localization.
\end{remark}

The compatibility conditions show that the 
quasideterminant with differential entries agrees with the
Brundan--Kleshchev generating matrix after the generators $E_{a,b}^{i,j}$
are replaced by their differential realizations. On the other hand, it
agrees with the reduction of $|z(u+\lambda-D)z^{-1}|_{n,1}$.
The latter can be computed directly in terms of the invariant coordinates
$\mathbf z_i$ and the operators $D_{i,j}$. 
This computation is independent
of the choice of localization: changing the localization does not alter
the identities used below or the resulting expression in the reduced
differential algebra. Accordingly, no particular construction of
$\mathcal{DN}$ or of the extended homomorphism is needed for the computation
below.

\begin{theorem}\label{th:intrinsic_factorization}
The Brundan--Kleshchev matrix, evaluated via the reduced differential
realization $\rho_\lambda^{\mathrm{red}}$ induced by $\rho_\lambda$ in
\eqref{eq:differential_realization}, satisfies the factorization

\begin{equation}\label{eq:differential_quasideterminant_factorization}(\xi\circ\rho_\lambda^{\mathrm{red}})(\mathbf T(u))=
\bigl(u+z_{[n]}(\boldsymbol\lambda_n-D_{[n]})z_{[n]}^{-1}\bigr)\cdot
\cdots\cdot
\bigl(u+z_{[1]}(\boldsymbol\lambda_1-D_{[1]})z_{[1]}^{-1}\bigr)..
\end{equation}
\end{theorem}

\begin{proof}
Recall that the quasideterminant of the Lax matrix in the differential
realization is
    $$|L(u)|_{n,1}=
\left|
\begin{pmatrix} (u+\boldsymbol{\lambda}_1-\mathbf{D}_1)\mathbf{z}^{-1}_1 & -\mathbf{D}_{1,2}&
-\mathbf{D}_{1,3}&\ldots& -\mathbf{D}_{1,n}\\
0&u+\boldsymbol{\lambda}_2-\mathbf{D}_2& -\mathbf{D}_{2,3}&\ldots& -\mathbf{D}_{2,n}\\
0&0&u+\boldsymbol{\lambda}_3-\mathbf{D}_3&\ldots& -\mathbf{D}_{3,n}\\
\vdots&\vdots&\ddots& \ddots& \vdots\\
0&0&\ldots&0&\mathbf{z}_n(u+\boldsymbol{\lambda}_n-\mathbf{D}_{n})
     \end{pmatrix}
\right|_{n,1} $$
To establish the explicit form after reduction, we note that the entries
of the matrix in this quasideterminant commute with $\mathfrak m_n$.
Indeed, the right-shift
    operators $D_{i,j}$ commute with the infinitesimal left action of
    $\mathfrak m_n$, while the diagonal coordinates are invariant under this
    action, by Proposition~\ref{prp:differential_reduction}.

Hence, we can reduce the entries of the matrix directly using
Proposition~\ref{prp:reduction_of_D}. The result is
$$|L(u)|_{n,1}\cdot \overline{1}= \left|
\begin{pmatrix} (u+\boldsymbol{\lambda}_1-\mathbf{D}_1)\mathbf{z}^{-1}_1 & \mathbf{z}_1^{-1} \mathbf{z}_2&
0&\ldots& 0\\
0&u+\boldsymbol{\lambda}_2-\mathbf{D}_2&  \mathbf{z}_2^{-1} \mathbf{z}_3&\ddots&\vdots\\
0&0&u+\boldsymbol{\lambda}_3-\mathbf{D}_3&\ddots& 0\\
\vdots&\vdots&\ddots& \ddots&  \mathbf{z}_{n-1}^{-1} \mathbf{z}_n\\
0&\ldots&0&0&\mathbf{z}_n(u+\boldsymbol{\lambda}_n-\mathbf{D}_{n})
 \end{pmatrix}
\right|_{n,1}.$$
The last quasideterminant can be computed explicitly using
Proposition~\ref{prp:nice_formulas_for_quasidet}. Thus, in the reduced
algebra,
$$|L(u)|_{n,1}\cdot\overline{1}=(-1)^{n-1}
\mathbf{z}_n(u+\boldsymbol{\lambda}_n-\mathbf{D}_{n})\mathbf{z}_n^{-1}
\cdot\ldots\cdot\mathbf{z}_1(u+\boldsymbol{\lambda}_1-\mathbf{D}_1)\mathbf{z}^{-1}_1\cdot\overline{1}.$$

On the other hand, by \eqref{eq:quasidet_BK_polynomial}, the
Brundan--Kleshchev matrix is $(-1)^{n-1}$ times the reduced
quasideterminant. Thus the sign in the last display cancels under this
normalization.

Next, we use the identification $\xi$ between the reduced algebra $\mathcal{D}_-$ and $\mathcal{D}(N_-')^{\otimes n}$ to obtain \eqref{eq:differential_quasideterminant_factorization}. In the original coordinates on $N_-$, the operators $\mathbf D_a$ may involve derivatives with respect to off-diagonal coordinates. Under the isomorphism $\xi:\mathcal D_-\xrightarrow{\sim}\mathcal D(N_-')^{\otimes n}$ from Proposition~\ref{prp:differential_reduction}, each factor is expressed in terms of differential operators acting in the individual tensor factors.
\end{proof}
Together, the Miura factorization \eqref{eq:Miura_map} and
Theorem~\ref{th:intrinsic_factorization} are
summarized by the commutative diagram
 \begin{equation}
    \label{eq:factorization_diagram}
    \begin{tikzcd}[column sep=large,row sep=large]
       \mathcal W_n(\mathfrak{gl}_p)
       \arrow[r, "\Theta"]
       \arrow[d, "\rho_\lambda^{\mathrm{red}}"']
       & \mathcal U(\mathfrak{gl}_p)^{\otimes n}
       \arrow[d, "{\rho_{\boldsymbol{\lambda}_1}\otimes\cdots\otimes
       \rho_{\boldsymbol{\lambda}_n}}"] \\
       \mathcal D_-
       \arrow[r, "\xi"']
       & \mathcal D(N_-')^{\otimes n}
    \end{tikzcd}
    \end{equation}
Each factor in
\eqref{eq:differential_quasideterminant_factorization} is the differential
realization of the corresponding factor $u+\mathbf E_i$ in the Miura
product. Consequently,
\begin{equation}\label{eq:Miura_differential_compatibility}
\xi\circ\rho_\lambda^{\mathrm{red}}
=\bigl(\rho_{\boldsymbol\lambda_1}\otimes\cdots\otimes
\rho_{\boldsymbol\lambda_n}\bigr)\circ\Theta.
\end{equation}
Thus the direct Hamiltonian reduction agrees with first applying the
quantum Miura homomorphism and then applying the tensor-product differential
realization.
\section{Conclusion}
In this work, we studied the Hamiltonian reduction of the big-cell
differential realization in the rectangular case. We showed that this
realization is naturally compatible with Hamiltonian reduction, yielding an
intrinsic factorization of the reduced differential realization of the
rectangular $W$-algebra.

We also want to emphasize the usefulness of quasideterminants for describing
invariants in the $W$-algebra and the reduced differential algebra. To
illustrate the geometric picture, recall the classical analogue of our
construction: the differential realization is replaced by the Springer map
from the cotangent bundle $T^*N_-$ to the coadjoint orbit
$\mathcal{O}_\lambda$ \cite{Babich,ChrissGinzburg}. Hamiltonian reduction of
this map gives Darboux coordinates on a symplectic leaf of the Slodowy slice,
which is quantized by the finite $W$-algebra \cite{GanGinzburg,Losev}.
In the rectangular case, this picture extends to a description of the symplectic leaves of the Yangian \cite{Shapiro}.

The main steps of our construction have classical geometric counterparts.
In the classical setting, the matrix entries commute, so the localization
and factorization procedures for the quasideterminant can be carried out
directly. The quasideterminant then reduces to the row determinant of the
Lax operator with constrained entries.

The natural properties of the invariant suggested by the geometric
construction also hold for the quasideterminant in the noncommutative,
quantum setting considered here. Thus, quasideterminants provide a natural
language for extending the classical picture to the quantum setting. They
also make the underlying algebraic structure transparent: the RTT relations
can be checked directly in this framework.

We would like to note that our analysis is restricted to the rectangular
reduction and to the differential realization considered above, which is
associated with a generic coadjoint orbit. More general Lax operators associated with special
orbits have also been studied \cite{FrassekPestun,FrassekPestunTsymbaliuk}.
However, when the differential realization is specialized to a concrete space, our construction applies
to a large class of Yangian representations and thus appears quite general.

The natural question is what obstacles arise in extending this picture. We outline several possible directions:
\begin{itemize}
    \item How can the factorization mechanism be extended to non-rectangular reductions?
    \item What obstructs generalization to more general type-A Lax operators? We suspect that difficulties arise only when the Lax operator's Jordan-block structure interacts nontrivially with the reduction.
    \item How can intrinsic factorization be extended to types B, C, and D and their associated $W$-algebras?
\end{itemize}
The first and third questions have been considered in related
quasideterminant constructions. For non-rectangular reductions in type A,
De Sole, Kac, and Valeri constructed a generating matrix using a generalized
quasideterminant \cite{DSKV}. This matrix does not by itself provide a full
set of generators for the $W$-algebra; its relation to a full generating set
was subsequently established in \cite{DSFVgen}. For classical simple Lie
algebras, generalized-quasideterminant generating matrices with coefficients
in the finite $W$-algebra were constructed in \cite{DSKVlax}. It remains
useful to understand how such generating matrices encode a full set of generators.
Our differential-realization approach, grounded in geometric reasoning, may
offer a complementary perspective and reveal a simpler structure for the
generators in the differential realization.

Finally, we emphasize that the present result provides a starting point for
studying the reduction of intertwiners and factorized $R$-matrices, which we
plan to investigate in future work.

\section*{Acknowledgments} 

We thank BIMSA for hospitality. 
A big part of the work was done during our visit to BIMSA. 
  
We thank N. Belousov, S. Khoroshkin, N. Reshetikhin, A. Litvinov, 
and V. Rubtsov for insightful discussions.
We would like to express our special thanks to Vitaly Tarasov.
Thanks to our joint discussions, all participants were able to 
appreciate the power and beauty of quasideterminants. 

The work of D.G. performed in Skoltech has been supported by the Russian Science Foundation under the grant 26-11-00342.

\appendix

\section{Commutation relations } \label{appA}

\subsection{Details of calculations in index form} \label{app:det}

Throughout this subsection, we use the summation convention: repeated block
indices are summed over $1$ through $n$, and repeated internal indices over $1$
through $p$, whether written as subscripts or superscripts.
We start from relation
\begin{align}\label{YA}
(u-v)\,L^{k,i}_{a,b}(u)\,L^{s,j}_{c,d}(v)+
L^{s,i}_{c,b}(u)\,L^{k,j}_{a,d}(v) =  
(u-v)\,L^{s,j}_{c,d}(v)\,L^{k,i}_{a,b}(u)+L^{s,i}_{c,b}(v)\,L^{k,j}_{a,d}(u),
\end{align}
where 
\begin{align}
L_{a,b}^{i,j}(u) = u \delta_{a,b}\delta^{i,j} + E_{b,a}^{j,i}.
\end{align}
Let us introduce the inverse matrix 
\begin{align}
L_{a,b}^{i,j}(u) M_{b,c}^{j,k}(u) = \delta_{a,c}\delta^{i,k}\ \ ;\
M_{a,b}^{i,j}(u) L_{b,c}^{j,k}(u) = \delta_{a,c}\delta^{i,k}.
\end{align}
Let us transform relation \eqref{YA} in a following way.
On the first step we multiply on the matrix $M(v)$ from the right
\begin{align*}
(u-v)\,L^{k,i}_{a,b}(u)\,L^{s,j}_{c,d}(v)\,M^{j,m}_{d,f}(v)+
L^{s,i}_{c,b}(u)\,L^{k,j}_{a,d}(v)\,M^{j,m}_{d,f}(v) \\
= (u-v)\,L^{s,j}_{c,d}(v)\,L^{k,i}_{a,b}(u)\,M^{j,m}_{d,f}(v)+
L^{s,i}_{c,b}(v)\,L^{k,j}_{a,d}(u)\,M^{j,m}_{d,f}(v),
\end{align*}
which gives 
\[
\begin{gathered}
(u-v)\,L^{k,i}_{a,b}(u)\,\delta^{s,m}\delta_{c,f}+
L^{s,i}_{c,b}(u)\,\delta^{k,m}\delta_{a,f} 
=(u-v)\,L^{s,j}_{c,d}(v)\,L^{k,i}_{a,b}(u)\,M^{j,m}_{d,f}(v) \\
+L^{s,i}_{c,b}(v)\,L^{k,j}_{a,d}(u)\,M^{j,m}_{d,f}(v),
\end{gathered}
\]
and on the second step multiply on $M(v)$
from the left 
\begin{align*}
(u-v)\,M^{q,s}_{e,c}(v)\,L^{k,i}_{a,b}(u)\,\delta^{s,m}\delta_{c,f}+
M^{q,s}_{e,c}(v)\,L^{s,i}_{c,b}(u)\,\delta^{k,m}\delta_{a,f}  \\=
(u-v)\,M^{q,s}_{e,c}(v)\,L^{s,j}_{c,d}(v)\,L^{k,i}_{a,b}(u)\,M^{j,m}_{d,f}(v)+
M^{q,s}_{e,c}(v)\,L^{s,i}_{c,b}(v)\,L^{k,j}_{a,d}(u)\,M^{j,m}_{d,f}(v),
\end{align*}
which gives
\begin{align*}
(u-v)\,M^{q,s}_{e,c}(v)\,L^{k,i}_{a,b}(u)\,\delta^{s,m}\delta_{c,f}+
M^{q,s}_{e,c}(v)\,L^{s,i}_{c,b}(u)\,\delta^{k,m}\delta_{a,f}  \\=
(u-v)\,\delta^{q,j}\delta_{e,d}\,L^{k,i}_{a,b}(u)\,M^{j,m}_{d,f}(v)+
\delta^{q,i}\delta_{e,b}\,L^{k,j}_{a,d}(u)\,M^{j,m}_{d,f}(v),
\end{align*}
and finally
\begin{equation}\label{MLA}
\begin{gathered}
(u-v)\,M^{q,m}_{e,f}(v)\,L^{k,i}_{a,b}(u) 
+M^{q,s}_{e,c}(v)\,L^{s,i}_{c,b}(u)\,\delta^{k,m}\delta_{a,f} \\
= (u-v)\,L^{k,i}_{a,b}(u)\,M^{q,m}_{e,f}(v) 
+\delta^{q,i}\delta_{e,b}\,L^{k,j}_{a,d}(u)\,M^{j,m}_{d,f}(v).
\end{gathered}
\end{equation}

To prove the Yangian commutation relations \eqref{TT} for the
quasideterminant's matrix elements, we continue transforming the main
relation \eqref{YA} and multiply \eqref{MLA} by $M(u)$ from the right.
\begin{align*}
(u-v)\,M^{q,m}_{e,f}(v)\,L^{k,i}_{a,b}(u)\,M^{i,j}_{b,h}(u) +
M^{q,s}_{e,c}(v)\,L^{s,i}_{c,b}(u)\,M^{i,j}_{b,h}(u)\,\delta^{k,m}\delta_{a,f}  \\=
(u-v)\,L^{k,i}_{a,b}(u)\,M^{q,m}_{e,f}(v)\,M^{i,j}_{b,h}(u)+
\delta^{q,i}\delta_{e,b}\,L^{k,\alpha}_{a,\beta}(u)\,M^{\alpha,m}_{\beta,f}(v)\,M^{i,j}_{b,h}(u),
\end{align*}
so that 
\begin{align*}
(u-v)&\,M^{q,m}_{e,f}(v)\,\delta^{k,j}\delta_{a,h} +
M^{q,s}_{e,c}(v)\,\delta^{s,j}\delta_{c,h}\,\delta^{k,m}\delta_{a,f} \\
&= (u-v)\,L^{k,i}_{a,b}(u)\,M^{q,m}_{e,f}(v)\,M^{i,j}_{b,h}(u)+
\delta^{q,i}\delta_{e,b}\,L^{k,\alpha}_{a,\beta}(u)\,M^{\alpha,m}_{\beta,f}(v)\,M^{i,j}_{b,h}(u),
\end{align*}
or 
\begin{align*}
(u-v)&\,M^{q,m}_{e,f}(v)\,\delta^{k,j}\delta_{a,h} +
M^{q,j}_{e,h}(v)\,\delta^{k,m}\delta_{a,f} \\ 
&= (u-v)\,L^{k,i}_{a,b}(u)\,M^{q,m}_{e,f}(v)\,M^{i,j}_{b,h}(u)+
L^{k,t}_{a,d}(u)\,M^{t,m}_{d,f}(v)\,M^{q,j}_{e,h}(u),
\end{align*}
and then multiply on the matrix $M(u)$ from the left 
\begin{align*}
&(u-v)\,M^{s,k}_{c,a}(u)\,M^{q,m}_{e,f}(v)\,\delta^{k,j}\delta_{a,h} +
M^{s,k}_{c,a}(u)\,M^{q,j}_{e,h}(v)\,\delta^{k,m}\delta_{a,f} \\
&= (u-v)\,M^{s,k}_{c,a}(u)\,L^{k,i}_{a,b}(u)\,M^{q,m}_{e,f}(v)\,M^{i,j}_{b,h}(u)+
M^{s,k}_{c,a}(u)\,L^{k,t}_{a,d}(u)\,M^{t,m}_{d,f}(v)\,M^{q,j}_{e,h}(u),
\end{align*}
which gives
\begin{align*}
(u-v)\,M^{s,j}_{c,h}(u)\,M^{q,m}_{e,f}(v) +
M^{s,m}_{c,f}(u)\,M^{q,j}_{e,h}(v) \\
= (u-v)\,M^{q,m}_{e,f}(v)\,M^{s,j}_{c,h}(u)+
M^{s,m}_{c,f}(v)\,M^{q,j}_{e,h}(u),
\end{align*}
or 
\begin{align}\label{MMA}
(u-v)\,M^{s,j}_{c,h}(u)\,M^{q,m}_{e,f}(v) +
M^{s,m}_{c,f}(u)\,M^{q,j}_{e,h}(v) =  
(u-v)\,M^{q,m}_{e,f}(v)\,M^{s,j}_{c,h}(u)+M^{s,m}_{c,f}(v)\,M^{q,j}_{e,h}(u).
\end{align}
We obtain the same Yangian relations as for $L$-matrices and 
the particular case $c=e=1\,,h=f=n$ gives Yangian relations 
for the quasideterminants
\begin{align*}
(u-v)\,M^{s,j}_{1,n}(u)\,M^{q,m}_{1,n}(v) +
M^{s,m}_{1,n}(u)\,M^{q,j}_{1,n}(v) =  
(u-v)\,M^{q,m}_{1,n}(v)\,M^{s,j}_{1,n}(u)+M^{s,m}_{1,n}(v)\,M^{q,j}_{1,n}(u).
\end{align*}
\section{Functional realization of the reduction}\label{appB}

\label{app:func}
Let us briefly comment on the functional side of the reduction described above. We postpone the detailed study to our second work; however, here we clarify the meaning of the reduction in a concrete realization on some functional space.

Similarly to the construction of the homomorphism $\rho_\lambda$, we start from the test functions on $N_-$ with coordinates $(m,z_{\mathrm{d}})$. For the purpose of reduction, we have to fix the behaviour of the test function with respect to the left action of $\mathfrak{m}_n$ (i.e. to consider the subspace of Whittaker vectors as a representation of our $W$-algebra or reduced algebra $\mathcal{D}_-$).

It is straightforward to check that the relation
$$
x\cdot \Phi(z)=\psi(x) \Phi(z), \quad \forall x\in \mathfrak{m}_n,
$$
implies
\begin{align*}
\Phi(m\,,z_{\mathrm d}) = e^{-\psi(\log m)}\Psi(z_{\mathrm d})\,,
\end{align*}
where $\Psi$ can be an arbitrary function of $z_\mathrm{d}$. For the usual choice of the character $\psi(E^{i,j}_{a,b})=\delta_{b,a-1}\delta_{i,j}$, the prefactor is $e^{-\psi(\log m)}=e^{-\sum_{k=1}^{n-1}\tr (\mathbf{m}_{k+1,k})}$, where $\mathbf{m}_{k+1,k}=\mathbf{z}_{k+1,k}\mathbf{z}_k^{-1}$ are blocks of $m$.

We now interpret the abstract map $\xi$ as an intertwining map between the subspace of functions $\Phi(m\,,z_{\mathrm d})$ with prefactor and the subspace of $\Psi(z_{\mathrm{d}})$, via conjugation by exponentials and setting the off-diagonal coordinate parts in differential operators to zero. Let us derive the explicit form of the reduction of the generators $\mathbf{D}_{a,b}$. To be explicit, we consider the $n=2$ case.
\begin{itemize}
\item The reduction of the operators $\mathbf{D}_a$ does not depend on the concrete choice of the prefactor and can be determined almost without calculations.
\end{itemize}
The subspace consists of the following functions:
\begin{align*}
\Phi(z) =  e^{-\psi(\mathbf{z}_{2,1}\mathbf{z}_1^{-1})}\, \Psi(\mathbf{z}_1\,,\mathbf{z}_2)\,.
\end{align*}
For the right shifts $D_a$ corresponding to $\mathbf{D}_1$ we have
\begin{align*}
\begin{pmatrix}\mathbf{z}_1&0\\
\mathbf{z}_{2,1}&\mathbf{z}_2\end{pmatrix}
\begin{pmatrix}1+\varepsilon \mathbf{a}&0\\
0 &1 \end{pmatrix} =
\begin{pmatrix}1&0\\
\mathbf{z}_{2,1}\mathbf{z}_1^{-1}&1\end{pmatrix}
\begin{pmatrix}\mathbf{z}_1\left(1+\varepsilon \mathbf{a}\right) &0\\
0&\mathbf{z}_2\end{pmatrix},
\end{align*}
so that
\begin{align*}
D_a \Phi(z) = e^{-\psi\left(\mathbf{z}_{2,1}\mathbf{z}_1^{-1}\right)\, }\left.\frac{d}{d\varepsilon}
\Psi(\mathbf{z}_1\left(1+\varepsilon \mathbf{a}\right)\,,\mathbf{z}_2)\right|_{\varepsilon=0}.
\end{align*}
In matrix form we have the intertwining relation
\begin{align*}
\mathbf{D}_{1}[ e^{-\psi(\mathbf{z}_{2,1}\mathbf{z}_1^{-1})}\Psi(z_{\mathrm{d}})] = e^{-\psi(\mathbf{z}_{2,1}\mathbf{z}_1^{-1})}[\left(\mathbf{d}_{1}\mathbf{z}_{1}\right)_{+}\Psi(z_{\mathrm{d}})],
\end{align*}
and a similar relation holds for $\mathbf{D}_2$.
It is instructive to compare with the initial picture without reduction.
In this case one works with functions $\Phi(\mathbf{z}_{2,1}\,,\mathbf{z}_{1}\,,\mathbf{z}_{2})$ and  
\begin{align}
\begin{pmatrix}\mathbf{z}_1&0\\
\mathbf{z}_{2,1}&\mathbf{z}_2\end{pmatrix}
\begin{pmatrix}1+\varepsilon \mathbf{a}&0\\
0 &1 \end{pmatrix} = 
\begin{pmatrix}\mathbf{z}_1\left(1+\varepsilon \mathbf{a}\right) &0\\
\mathbf{z}_{2,1} \left(1+\varepsilon \mathbf{a}\right)&\mathbf{z}_2\end{pmatrix} 
\end{align}
so that $\mathbf{D}_{1}$ contains contributions from all the variables:
\begin{align*}
\mathbf{D}_{1} =
\left(\mathbf{d}_{1}\mathbf{z}_{1}\right)_{+} +
\left(\mathbf{d}_{1,2}\mathbf{z}_{2,1}\right)_{+}
\end{align*}
On a formal level $\mathbf{D}_{1}$ commutes with $e^{-\psi(\mathbf{z}_{2,1}\mathbf{z}_1^{-1})}$ so that we have 
\begin{align*}
\mathbf{D}_{1}[ e^{-\psi(\mathbf{z}_{2,1}\mathbf{z}_1^{-1})}\Psi(z_{\mathrm{d}})] = 
e^{-\psi(\mathbf{z}_{2,1}\mathbf{z}_1^{-1})}
[\left(\mathbf{d}_{1}\mathbf{z}_{1}\right)_{+} + \left(\mathbf{d}_{1,2}\mathbf{z}_{2,1}\right)_{+}]\Psi(z_{\mathrm{d}}) = 
e^{-\psi(\mathbf{z}_{2,1}\mathbf{z}_1^{-1})}
[\left(\mathbf{d}_{1}\mathbf{z}_{1}\right)_{+}]\Psi(z_{\mathrm{d}}).
\end{align*}
If we additionally relabel the variables $\mathbf{z}_i\rightarrow z_{[i]}$, then we find that the formal map $\xi$ implements exactly the passage from the subspace of functions $e^{-\psi(\log m)}\Psi(z_{[1]},z_{[2]})$ to the subspace $\Psi(z_{[1]},z_{[2]})$.

\begin{itemize}
\item The reduction of the operators $\mathbf{D}_{a,b}$ depends on the concrete choice of the character $e^{-\psi(\log m)}$.
\end{itemize}

For the left shifts $E_a$ corresponding to $\mathbf{E}_{1,2}$, we have
\begin{align*}
e^{-\varepsilon a} z = \begin{pmatrix}1&0\\
-\varepsilon\mathbf{a}&1\end{pmatrix} \begin{pmatrix}\mathbf{z}_1&0\\
\mathbf{z}_{2,1}&\mathbf{z}_2\end{pmatrix} =
\begin{pmatrix}1&0\\
\mathbf{z}_{2,1}\mathbf{z}_1^{-1}-\varepsilon\mathbf{a} &1\end{pmatrix}
\begin{pmatrix}\mathbf{z}_1&0\\
0&\mathbf{z}_2\end{pmatrix} \ \ ;\ \ a = \begin{pmatrix}1&0\\
\mathbf{a}&1\end{pmatrix} \in \exp (\mathfrak{m}_2),
\end{align*}
and therefore
\begin{align*}
E_a \Phi(z)
=\left.\frac{d}{d\varepsilon}e^{\psi(\mathbf{z}_{2,1}\mathbf{z}_1^{-1}-\varepsilon\mathbf{a})}
\Psi(\mathbf{z}_1\,,\mathbf{z}_2)\right|_{\varepsilon=0}.
\end{align*}
We consider the reduction associated with the character $\psi({E}_{2,1}^{i,j}) =\delta_{i,j}$, and this defines the prefactor
$e^{-\tr (\mathbf{z}_{2,1}\mathbf{z}_1^{-1})}$.
For the right shifts $D_a$ corresponding to $\mathbf{D}_{12}$ we have
\begin{align*}
\begin{pmatrix}\mathbf{z}_1&0\\
\mathbf{z}_{2,1}&\mathbf{z}_2\end{pmatrix}
\begin{pmatrix}1&0\\
\varepsilon \mathbf{a} &1 \end{pmatrix} =
\begin{pmatrix}1&0\\
\mathbf{z}_{2,1}\mathbf{z}_1^{-1} +\varepsilon
\mathbf{z}_2 \mathbf{a} \mathbf{z}_1^{-1}&1\end{pmatrix}
\begin{pmatrix}\mathbf{z}_1 & 0\\
0&\mathbf{z}_2\end{pmatrix},
\end{align*}
so that
\begin{align*}
D_a [ e^{-\tr (\mathbf{z}_{2,1}\mathbf{z}_1^{-1})}\Psi(\mathbf{z}_1,\mathbf{z}_2)]
&=
\left.\frac{d}{d\varepsilon} e^{-\varepsilon
\tr(\mathbf{z}_2 \mathbf{a} \mathbf{z}_1^{-1})}
\right|_{\varepsilon=0}\,e^{-\tr (\mathbf{z}_{2,1}\mathbf{z}_1^{-1})}
\Psi(\mathbf{z}_1\,,\mathbf{z}_2)\\&= e^{-\tr (\mathbf{z}_{2,1}\mathbf{z}_1^{-1})}[-\tr(\mathbf{a} \mathbf{z}_1^{-1}\mathbf{z}_2)
\Psi(\mathbf{z}_1\,,\mathbf{z}_2)].
\end{align*}
This formula gives an intertwining relation
\begin{align*}
\mathbf{D}_{1,2}[ e^{-\tr (\mathbf{z}_{2,1}\mathbf{z}_1^{-1})}\Psi(\mathbf{z}_1,\mathbf{z}_2)] = e^{-\tr (\mathbf{z}_{2,1}\mathbf{z}_1^{-1})}[- \mathbf{z}_1^{-1}\mathbf{z}_2
\Psi(\mathbf{z}_1\,,\mathbf{z}_2)].
\end{align*}

\end{document}